\documentclass[12pt]{article}

\usepackage{amsmath,amssymb,amsthm,fullpage,charter,verbatim,calc}
\usepackage{pifont,amsfonts,amsmath,amstext,amsthm,amssymb}
\usepackage{graphicx,color}
\usepackage{vmargin}
\usepackage{mathtools}
\usepackage{MnSymbol}
\setpapersize{A4}
\setmarginsrb{30mm}{20mm}{30mm}{20mm}{12pt}{11mm}{0pt}{11mm}
\usepackage{epic,eepic}
\newtheorem{theorem}{Theorem}

\newtheorem{corollary}{Corollary}

\newtheorem{fact}{Fact}
\newtheorem{definition}{Definition}

\newcommand*{\clP}{$ \mathrm{\bf P} $}
\newcommand*{\clNP}{$ \mathrm{\bf NP} $}
\newcommand*{\cH}{\mathcal{H}}
\newcommand*{\cS}{\mathcal{S}}
\newcommand*{\clAPX}{$\mathrm{\bf APX}$}
\title{\bf Improved Lower Bounds for the Shortest Superstring and Related Problems\\[1ex]}
\author{
Marek Karpinski\thanks{Dept. of Computer Science and the Hausdorff
    Center for Mathematics, University of Bonn.
    Supported in part by DFG grants and the Hausdorff Center grant EXC59-1.
    Email:~\texttt{marek@cs.uni-bonn.de}}
\and    
    Richard Schmied\thanks{Dept. of Computer Science, University of Bonn.
    Work supported by Hausdorff Doctoral Fellowship.
    Email:~\texttt{schmied@cs.uni-bonn.de}}
}
\date{}

\begin{document}
\maketitle

\begin{abstract}
We study the approximation hardness of the
Shortest Superstring, the Maximal Compression  and the 
Maximum Asymmetric Traveling Salesperson 
(MAX-ATSP) problem. 
We introduce a new reduction method that produces strongly restricted instances of the Shortest 
Superstring problem, in which the maximal orbit size is eight
 (with no character
appearing more than eight times) and all given strings having length at most 
four. Based on this reduction method, we are able to improve 
the best up to now known approximation lower bound 
for the Shortest Superstring problem and the Maximal Compression problem
by an order of magnitude.
The results imply  also  an improved approximation lower bound 
for the MAX-ATSP problem. 
\end{abstract}

\section{Introduction}
In the \textbf{Shortest Superstring problem}, we are given
a finite set $S$ of strings and we would like to construct their shortest 
superstring, which is  the shortest possible string such that every
string in $S$ is a proper substring of it.

The task of computing a shortest common superstring appears in a wide variety of application 
 related to computational biology (see. e.g.~\cite{L88} and \cite{L90}). 
Intuitively,  short superstrings preserve important
biological structure and are good models of the original DNA sequence. 
In context of computational biology, DNA sequencing is
the important task of determining the sequence of nucleotides in a molecule of DNA.
The DNA can be seen as a double-stranded
sequence of four types of nucleotides
 represented by the alphabet $\{a,c,g,t\}$.
 Identifying those strings
for different molecules is an important
step towards understanding the biological functions of the
molecules. However, with current laboratory methods, it is quite
 impossible to extract a long molecule directly
as a whole. In fact, biochemists split
 millions of  identical molecules into pieces each
typically containing at most $500$ nucleotides.
Then, from  sometimes
millions of these fragments, one has to compute the superstring representing
the whole molecule.

From the computational point of view, the Shortest Superstring problem
is an optimization problem, which consists of finding a minimum length superstring
for a given set $S$ of strings over a finite alphabet $\Sigma$.
The underlying decision version was proved to be
 \clNP-complete~\cite{MS77}.
 However, there are many applications that involve
 relatively simple classes of strings.
Motivated by those applications, many authors have investigated 
whether the Shortest Superstring problem becomes polynomial time
solvable under various restrictions to the set of instances.
Gallant et al.~\cite{GMS80} proved that this problem in the exact setting 
is still \clNP-complete for strings of length three and polynomial time solvable 
for strings of length two. 
On the other hand, Timkovskii~\cite{T90} studied the 
Shortest Superstring problem under restrictions to the orbit size
of the letters in $\Sigma$.  The orbit size of
a letter is the number of its occurrences in the strings of $S$. 
Timkovskii proved that this problem restricted to instances with
maximal orbit size two is polynomial time solvable. He 
raised the question about the status of the problem
 with maximal orbit size $k$  for any constant $k\geq 3$. 
It is known that the Shortest Superstring problem  remains
\clNP-hard for the following strongly restricted instances, such as
\begin{enumerate}
\item[$(i)$]
all strings have length four and the maximal orbit size is six~\cite{M94},
\item[$(ii)$]
the size of the alphabet of the instance  is exactly two~\cite{GMS80}, and
\item[$(iii)$]
 all strings are of the form $10p10q$, where $p,q\in \mathbb{N}$~\cite{M98}.
\end{enumerate} 
In order to cope with the exact computation intractability, approximation algorithms were 
designed
to deal with this problem. The first polynomial time approximation
algorithm with a constant approximation ratio was given by Blum et al.~\cite{BJL$^+$94}.
It achieves an  approximation ratio  $3$. This factor was improved in a  series
of papers  yielding approximation ratios of $2.88$ by Teng and Yao~\cite{TT93}; $2.83$
by Czumaj et al.~\cite{CGP$^+$94}; $2.79$
 by Kosaraju, Park, and Stein~\cite{KPS94}; $2.75$ by Armen and Stein~\cite{AS95};
 $2.67$ by Armen and Stein~\cite{AS98}; $2.596$ by Breslauer, Jiang, and Jiang~\cite{BJJ97}
and $2.5$ by Sweedyk~\cite{S99}. 
The currently best known  approximation algorithm is due to Mucha~\cite{M12}
and yields an approximation ratio of $2.478$.

On the lower bound side, Blum et al.~\cite{BJL$^+$94} proved that approximating the
Shortest Superstring problem is \clAPX-hard.
However, the constructed reduction  produces instances
with arbitrarily large alphabets. In~\cite{O99}, Ott provided the first
explicit approximation hardness result and 
proved that the problem is
\clAPX-hard even if the size of the alphabet is two.
In fact, Ott proved that instances over a binary alphabet
are \clNP-hard to approximate with an approximation ratio $17246/17245~(1.000057)-\epsilon $ for every $\epsilon > 0$. 
In 2005, Vassilevska~\cite{V05} gave an improved approximation lower bound of $1217/1216~(1.00082)$ by using a natural construction.
 The constructed instances  of the Shortest
Superstring problem have maximal orbit size $20$ and the length of the strings is  exactly $4$.
 
 In this paper, we prove that even instances of the
Shortest Superstring problem with maximal orbit size $8$ and all strings having length $4$ 
are \clNP-hard
to approximate with less than $333/332~(1.00301)$.\\
\\
 \textbf{Maximal Compression problem}. We are given a collection
of strings $S=\{s_1,\ldots,s_n\}$. The task is  to find a superstring for $S$
with maximum compression, which is the difference between the sum of
the lengths of the given strings and the length of the superstring. 

In the exact setting, an optimal solution to the Shortest Superstring problem
is an optimal solution to this problem,
but the approximate solutions can differ significantly
in the sense of approximation ratio.
The Maximal Compression problem arises in
various data compression problems (cf.~\cite{SS82},~\cite{S88} and~\cite{MJ75}).
The decision version of this problem
 is \clNP-complete~\cite{MS77}.
Tarhio and Ukkonen~\cite{TU88} and Turner~\cite{T89}
gave approximation algorithms with approximation ratio $2$. 
The best known approximation upper bound  is $1.5$~\cite{KLS$^+$05}
by reducing it to the MAX-ATSP problem, which is defined below.

On the approximation lower bound side, Blum et al.~\cite{BJL$^+$94} proved the 
\clAPX-hardness of the Maximal Compression problem. 
The first explicit approximation lower bounds were given by Ott~\cite{O99},   
who proved that it is \clNP-hard to approximate this problem with an approximation factor 
$11217/11216~(1.000089)-\epsilon $ for every $\epsilon>0$.  
This hardness result was improved by Vassilevska~\cite{V05} implying a lower bound of
$1072/1071~(1.00093)-\epsilon~ $ for any $\epsilon>0$, unless \clP~$=$~\clNP. In this paper, we prove that
approximating the Maximal Compression problem
 with an approximation ratio less than $204/203 ~(1.00492)$ is \clNP-hard.\\
\\
\textbf{Maximum Asymmetric Traveling Salesperson (MAX-ATSP) problem}.
We are given a complete directed graph $G$  and a weight function
$w$ assigning each edge of $G$ a nonnegative weight. The task is to find a closed
tour of maximum weight visiting every vertex of $G$ exactly once .

This problem has various applications and in fact, a good approximation algorithm for MAX-ATSP
yields a good approximation algorithm for many other optimization problems
such as the  Shortest Superstring problem, the Maximum Compression problem
and the Minimum Asymmetric $(1,2)$-Traveling Salesperson (MIN-$(1,2)$-ATSP) problem. 
The latter problem is the restricted version
of the Minimum Asymmetric Traveling Salesperson problem, in which 
we restrict the weight function $w$ to weights one and two.
The MAX-ATSP problem can be seen as a generalization of the MIN-$(1,2)$-ATSP
problem in the sense that any $(\frac{1}{\alpha})$-approximation algorithm for the former 
problem transforms in a $(2-\alpha)$- approximation algorithm for the latter problem.
Due to this reduction, all negative results concerning the approximation of the MIN-$(1,2)$-ATSP
problem imply hardness results for the MAX-ATSP problem. 
Since MIN-$(1,2)$-ATSP  is \clAPX-hard~\cite{PY93}, there is little hope for polynomial time
approximation algorithms with arbitrary good precision for the MAX-ATSP problem.  
On the other hand, the first approximation algorithm for the MAX-ATSP problem with guaranteed
approximation performance is due to Fisher, Nemhauser, and Wolsey~\cite{FNW79} and achieves
an approximation factor of $2$.
After that Kosaraju, Park, and Stein~\cite{KPS94} gave an approximation algorithm for that problem
with performance ratio $1.66$.
This result was improved by Bl\"aser~\cite{B02} who obtained an approximation upper bound of
 $1.63$. Lewenstein and Sviridenko~\cite{LS03} were able to improve the 
 approximation upper bound for that problem to $1.60$. 
Then, Kaplan et al.~\cite{KLS$^+$05}  designed an algorithm for the MAX-ATSP problem
  yielding the best known approximation upper bound of $1.50$.
  
On the approximation hardness side, Engebretsen~\cite{E99} proved that, for any $\epsilon>0$,
there is no ($2805/2804-\epsilon$)-approximation algorithm for MIN-$(1,2)$-ATSP, unless \clP~ $=$~ \clNP,
which yields an approximation lower bound of $2804/2803~ (1.00035)-\epsilon$ for the MAX-ATSP problem. 
The negative result was improved by Engebretsen and Karpinski~\cite{EK06}
to $321/320~(1.0031)-\epsilon $ for the MIN-$(1,2)$-ATSP problem. It implies the best known approximation
lower bound of $320/319~(1.0031 )-\epsilon$, unless \clP~ $=$~ \clNP. In this paper, we prove that
approximating the MAX-ATSP problem with an approximation ratio less than $204/203 ~(1.00492)$ is \clNP-hard.
\section{Preliminaries}
In the following,  we 
introduce some notation and abbreviations.\\
\\
Throughout, for $i\in \mathbb{N}$, we use the abbreviation $[i]$ for the set $\{1,\ldots,i\}$.
Given an finite alphabet $\Sigma$, a string is an element of $\Sigma^*$.
Given two strings $v=v_1\cdots v_n$ and $w=w_1\cdots w_m$ over 
$\Sigma$, we denote the length of $v$  by $|v|$.
Furthermore, $v$ is a substring of $w$, if $m \ge n$
and there exists a $j\in \{0,..,n-m\}$ such that for all $i\in [m]$,
$v_i=w_{j+i}$. $w$ is said to be a superstring of $v$ if
$v$ is a substring of $w$. Given a set of strings
$S = \{s_1, . . . , s_n\}\subset \Sigma^* $, a string $s\in \Sigma^*$ is a superstring for $S$
if $s$ is a  superstring of every $s_i\in S$. Given a superstring $s$ for $S$,  
the compression of $s$, denoted  $comp(S,s)$, is  defined as
$$comp(S,s)=\sum\limits_{s_i\in S}|s_i|-|s|.$$ In addition, we introduce the notion of 
the maximal orbit size
of $S$ being the maximal number of occurences of a character in $S$.
We are ready to give the definition of the Shortest Superstring problem
and the Maximal Compression problem.
\begin{definition}
Given an alphabet $\Sigma$ and a set of strings 
$S = \{s_1, . . . , s_n\}\subset \Sigma^* $
such that no string in $S$ is a substring of another string 
in $S$, in the Shortest Superstring problem we have to find a 
string $s$ for $S$ of minimum length, 
whereas
in the Maximum Compression problem, we have to find a superstring $s$
for $S$ with maximum 
compression.
\end{definition}
\noindent
In the following, we concentrate on the traveling salesperson problems.
We begin with the definition of the MAX-ATSP problem. For this reason,
we introduce the notion of a \emph{Hamiltonian tour}. Given a directed graph
$G=(V,A)$, a Hamiltonian tour is a cycle in $G$ visiting each vertex of
$G$ exactly once.
\begin{definition}[MAX-ATSP]
Given a complete directed graph $G=(V,A)$ and a weight function 
$w$ assigning each edge of $G$ a nonnegative weight,
the MAX-ATSP problem consists of finding a Hamiltonian tour of maximum weight in $G$. 
\end{definition}
\noindent 
Next, we give the definition of the  MIN-$(1,2)$-ATSP problem, which is closely related to
the MAX-ATSP problem.
\begin{definition}[MIN-$(1,2)$-ATSP]
In the MIN-$(1,2)$-ATSP problem, we are given a complete directed 
graph $G=(V,A)$ and a weight function $w:A\rightarrow \{1,2\}$. The task is to find
a Hamiltonian tour of minimum weight in $G$.
\end{definition}
\section{Related Work}
In the following, we present some results related to the problems studied in this paper. 
In particular, we describe briefly some reductions, which we use later on.\\
\\ 
The following theorem is due to Vassilevska~\cite{V05} and deals with 
best known approximation lower bounds for the Shortest Superstring problem as well as
for the Maximal Compression problem. 
\begin{theorem}[\cite{V05}]
For any $\epsilon > 0$, it is  \clNP-hard to approximate the  Shortest Superstring problem 
and the Maximal Compression problem restricted to instances 
with equal length strings  in polynomial time within a factor
of
\begin{itemize}
\item $1.00082-\epsilon$  and
\item  $1.00093-\epsilon$,   respectively.
\end{itemize}
\end{theorem}
\noindent 
In addition, the maximal orbit size of the constructed instances in~\cite{V05} is
 $20$ and all strings have length four.
  In the same paper, it was proved that  the Shortest Superstring
problem is the hardest to approximate on instances over a binary alphabet.
\begin{theorem}[\cite{V05}]
Suppose the Shortest Superstring problem can be approximated by
a factor $\alpha$ on instances over a binary alphabet. Then, the
Shortest Superstring problem can be approximated by a factor $\alpha$
on instances over any alphabet.
\end{theorem}
\noindent
Given an instance $S$ of the Shortest Superstring problem,
consider the associated weighted complete graph, in which the vertices
are represented by the strings in $S$ and the weight
of an edge is given by the the number of maximum overlapped letters of the corresponding
strings. Then,  the
optimal compression is equivalent to the weight of
a maximum Hamiltonian path. By introducing a
 special vertex representing the start and the end
of the Hamiltonian cycle, the Maximal Compression
problem is equivalent to the MAX-ATSP problem on this graph.
This fact was used in~\cite{KLS$^+$05}  
in order to obtain an improved approximation
algorithm for the Maximal Compression problem. 
\begin{fact}\label{redmaxtspmaxcomp}
An $\alpha$-approximation algorithm for the MAX-ATSP problem implies an 
$\alpha$-approximation algorithm  for the Maximal Compression problem.
\end{fact}
\noindent
Another interesting relation can be derived by replacing all edges with weight
two of an instance of the MIN-$(1,2)$-ATSP problem by edges of weight zero and
then, computing a Hamiltonian tour of maximum weight.
Vishwanathan\cite{V92} proved that this transformation relates the   
MIN-$(1,2)$-ATSP problem to the MAX-ATSP problem in the following sense. 
%
%
\begin{theorem}[\cite{V92}]\label{sspredminmax}
An $(\frac 1{\alpha})$- approximation algorithm for the MAX-ATSP problem implies
an $(2-\alpha)$- approximation algorithm for the MIN-$(1,2)$-ATSP problem. 
\end{theorem}   
\noindent
Due to this reduction, every hardness result concerning the MIN-$(1,2)$-ATSP problem
can be transformed into a hardness result for the MAX-ATSP problem.
The best known approximation lower bound for the MIN-$(1,2)$-ATSP problem is 
proved in~\cite{EK06}
and it yields the following hardness result .
\begin{theorem}[\cite{EK06}]
For any constant $\epsilon > 0$, it is \clNP-hard to approximate 
the MIN-$(1,2)$-ATSP problem with an approximation ratio $1.0031-\epsilon$.
\end{theorem}
\noindent
According to Theorem~\ref{sspredminmax}, it implies the  hardness result for
the MAX-ATSP problem stated below. 
\begin{corollary}
For any constant $\epsilon > 0$, it is \clNP-hard to approximate 
the MAX-ATSP problem within $1.0031-\epsilon$.
\end{corollary}
\subsection{Hybrid Problem}

In their paper on approximation hardness of bounded occurrence instances of several
combinatorial optimization problems, Berman and Karpinski~\cite{BK99} introduced the 
Hybrid problem and proved that this problem is \clNP-hard to approximate
with some constant. 

\begin{definition}[Hybrid problem] 
Given a system of linear
equations mod 2 containing n variables, $m_2$ equations with exactly two variables, and $m_3$
equations with exactly three variables, find an assignment to the variables that satisfies as
many equations as possible.
\end{definition}
\noindent
In the aforementioned paper, Berman and Karpinski  proved the following
hardness result.
\begin{theorem}[\cite{BK99}]\label{ssphybridsatz}
For any constant $\epsilon > 0$, there exists instances of the Hybrid problem with $42\nu$
variables, $60\nu$ equations with exactly two variables, and $2\nu$ equations with exactly three
variables such that:
\begin{enumerate}
\item[$(i)$] Each variable occurs exactly three times.
\item[$(ii)$] Either there is an assignment to the variables that leaves at most $\epsilon\nu$ equations unsatisfied,
or else every assignment to the variables leaves at least $(1-\epsilon)\nu$ equations
unsatisfied.
\item[$(iii)$] It is \clNP-hard to decide which of the two cases in item $(ii)$ above holds.
\end{enumerate}
\end{theorem}
\noindent
Analyzing the details of their construction, it can be seen that
every instance of the Hybrid problem produced by it has an even more special structure. The equations
containing three variables are of the form $x \oplus y \oplus z = \{0, 1\}$. 
Those equations arise from the Theorem of H{\aa}stad~\cite{H01} concerning the hardness
of approximating equations with exactly three variables called the MAX-E3-LIN problem,
which can be seen as a special instance of the Hybrid problem.\\
\begin{figure}[h]
\begin{center}
\fbox{ 
\input{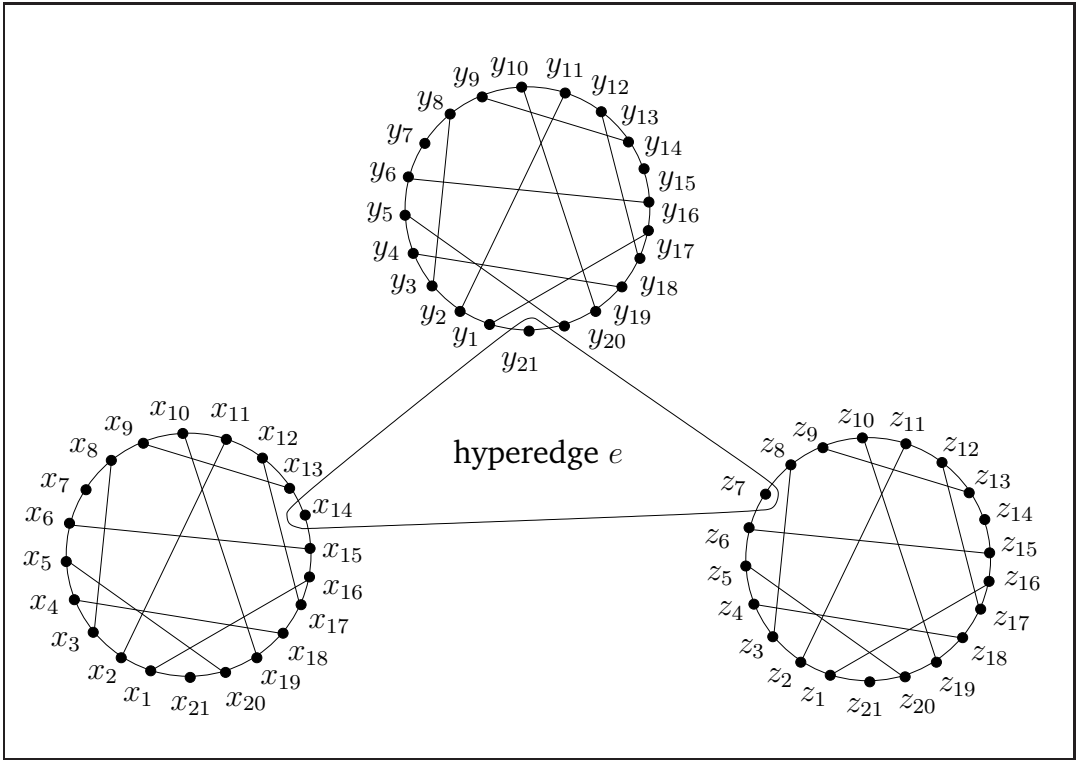}
}
\end{center}
\caption{ An example  of a Hybrid instance with circles 
$C^x$, $C^y$, $C^z$, and hyperedge $e=\{z_7,y_{21},x_{14}\}$.}
\label{fig:hybridssp}
\end{figure} 
\\
For every variable $x$ of the original instance $\mathcal{E}_3$ of the MAX-E3-LIN problem, they introduced a corresponding set  
of variables $V_x$. If the variable $x$ occurs $t_x$ times in $\mathcal{E}_3$, then, $V_x$ contains
$7t_x$ variables $x_1,\ldots, x_{7t_x}$. Furthermore, the variables in $V_x$ are connected 
by equations of the form $x_i\oplus x_{i+1}=0$ with $i\in [7t_x-1]$ and $x_1\oplus x_{7t_x}=0$.
This construction induces the circle $C^x$ on the variables $V_x$. In addition to it,
every circle $C^x$ possesses an associated matching $M^x$. The variables contained in
$Con(V_x)=\{x_i \mid i\in \{7\nu\mid \nu \in [t_x]\}  \}$ are called \emph{contact variables}, whereas
the variables in $V_x\backslash Con(V_x)$ are called \emph{checker variables}.

 Let $\mathcal{E}_3$ be an instance of the MAX-E3-LIN problem and $\cH$ be its corresponding instance of the
Hybrid problem. We denote by $V(\mathcal{E}_3)$ the set of variables which occur in the instance $\mathcal{E}_3$. 
Then, $\cH$ can be represented graphically by $|V(\mathcal{E}_3)|$ circles $C^x$ with $x\in V(\mathcal{E}_3)$ containing
the variables $V(C^x)=\{x_1, \ldots,x_{n^x}\}$ as vertices. The edges are identified
by the equations  included in  $\cH$. The equations with exactly three variables are represented by
hyperedges $e$ with cardinality $|e|=3$. 
The equations 
 $x_i\oplus x_{i+1}=0$ induce a circle containing 
 the vertices $\{x_1, \ldots,x_{n^x}\}$ and the matching equations $x_i\oplus x_{j}=0$ with $\{i,j\}\in M^x$
  induce a perfect matching on the set of checker variables. 
  An example of an instance of the Hybrid problem is depicted in Figure\ref{fig:hybridssp}.
   
  In summary, we notice that
 there are four type of equations in 
the Hybrid problem $(i)$ the circle equations $x_i\oplus x_{i+1}=0$ with $i\in [7t_x-1]$,  
$(ii)$ circle border equations $x_1\oplus x_{7t_x}$, $(iii)$ matching equations
$x_i\oplus x_j=0$ with $\{i,j\}\in M^x$, and $(iv)$ equations with three variables 
of the form $x\oplus y \oplus z=\{0,1\}$.\\
In the remainder, we  assume that equations of the form 
 $x \oplus y \oplus z = \{0, 1\}$ contain only unnegated variables due to
 the transformation $\bar{x} \oplus y \oplus z = 0 \equiv x \oplus y \oplus z = 1$.

\section{Our Contribution}
We now formulate our main result.
\begin{theorem}\label{hauptsatzssp}
Given an instance $\cH$ of the Hybrid problem  with $n$ circles,
$m_2$ equations with two variables and  $m_3$ equations with exactly
three variables with the properties described in Theorem~\ref{ssphybridsatz},
we construct in polynomial time an instance $\cS_{\cH}$ of
the Shortest Superstring problem and Maximal Compression problem with the following properties: 
\begin{enumerate}
\item[$(i)$] If there exists an assignment $\phi$ to the variables of $\cH$ which leaves at most  $u$ equations
unsatisfied, then, there exist a superstring $s_{\phi}$ for $\cS_{\cH}$  with  length
at most $|s_{\phi}|=5m_2+16 m_3+7n+u$.   
\item[$(ii)$] From every superstring $s$ for $S_{\cH}$  with length $|s|=5m_2+16 m_3+u+7n$,
we can construct in polynomial time  
an assignment $\psi_s$ to the variables
of $\cH$ that leaves at most $u$ equations in $\cH$ unsatisfied.
\item[$(iii)$] If there exists an assignment $\phi$ to the variables of $\cH$ which leaves at most  $u$ equations
unsatisfied, then, there exist a superstring $s_{\phi}$ for $\cS_{\cH}$  with  compression 
at least $comp(S_{\cH},s_{\phi})=3m_2+ 12 m_3-u+5n$ .
\item[$(iv)$] From every superstring $s$ for $S_{\cH}$  with  compression  
$comp(S_{\cH},s)=3m_2+ 12 m_3-u+5n$,
we can construct in polynomial time  
an assignment $\psi_s$ to the variables
of $\cH$ that leaves at most $u$ equations in $\cH$ unsatisfied.   
\item[$(v)$] The maximal orbit size of the instance $\cS_{\cH}$ is $8$ and the length 
 of each string in $\cS_{\cH}$ is  $4$.
\end{enumerate}
\end{theorem}
\noindent
The former theorem can be used to derive an explicit approximation lower bound for
the Shortest Superstring  problem by reducing instances of the Hybrid problem
of the form described in Theorem~\ref{ssphybridsatz} to  the Shortest Superstring  problem.
\begin{corollary}
For every $\epsilon >0$, it is  \clNP-hard to approximate the  Shortest Superstring  problem
with an approximation factor $\frac{333}{332}~( 1.00301)-\epsilon$.
\end{corollary}
\begin{proof}
First of all, we choose constants $k\in \mathbb{N}$ and  $\delta >0$ such that  
$\frac{333-\delta}{332+\delta +\frac {42}k}\geq
\frac{333}{332}-\epsilon$ holds. Given an instance $\mathcal{E}_3$ of the MAX-E3-LIN problem,
we generate $k$ copies of $\mathcal{E}_3$ and produce an instance 
 $\cH$ of the Hybrid problem. Then, 
we construct the corresponding instance $\cS_{\cH}$ of the Shortest Superstring problem
 with the properties described in Theorem~\ref{hauptsatzssp}.
We conclude according to Theorem~\ref{ssphybridsatz} that there exist a superstring
for $\cS_{\cH}$ 
with length at most 
$$5\cdot60\nu k+16\cdot2\nu k+\delta\nu k+ 7n \leq(332+\delta+\frac {7n}{k\nu} )\nu k
 \leq(332+\delta+\frac {7\cdot 6}k )\nu k$$ 
or the length of a superstring for $\cS_{\cH}$  is bounded from below by  
$$5\cdot60\nu k+16\cdot2\nu k+(1-\delta)\nu k+ 7n \geq (332+(1-\delta) )\nu k \geq (333-\delta)\nu k .$$
From Theorem~\ref{ssphybridsatz}, we know that the two cases above are \clNP-hard
to distinguish.
Hence, for every $\epsilon>0$, it is  \clNP-hard to find a solution to
 the  Shortest Superstring problem with an approximation ratio
 $\frac{333-\delta}{332+\delta +\frac {42}k}\geq
\frac{333}{332}-\epsilon$.
\end{proof}
\noindent
Analogously, Theorem~\ref{hauptsatzssp} can be used to derive an approximation lower bound for the Maximal Compression
problem. 
\begin{corollary}
For every $\epsilon >0$, it is  \clNP-hard to approximate the  Maximal Compression  problem
with an approximation factor $204/203~( 1.00492)-\epsilon$.
\end{corollary}
\noindent
By applying Fact~\ref{redmaxtspmaxcomp}, we obtain the following hardness result for the 
MAX-ATSP problem.
\begin{corollary}
For every $\epsilon >0$, it is  \clNP-hard to approximate the  MAX-ATSP  problem
with an approximation factor $204/203~( 1.00492)-\epsilon$.
\end{corollary}
%
\section{Reduction from the Hybrid Problem}
In this section, we present the proof of 
 Theorem~\ref{hauptsatzssp}. 
But before we describe the approximation preserving reduction, we first 
prove a slightly weaker result. The first approach uses 
strings with length $6$ simulating equations with exactly three variables. 
In section~\ref{sec:improvapproach}, we 
will introduce smaller gadgets for equations with exactly three variables implying
the claimed inapproximability results.\\
\\  
Let us state the properties of our first approach. 
\begin{theorem}\label{hauptsatzsspI}
Given an instance $\cH$ of the Hybrid problem  with $n$ circles,
$m_2$ equations with two variables and  $m_3$ equations with exactly
three variables with the properties described in Theorem~\ref{ssphybridsatz},
we construct in polynomial time an instance $\cS_{\cH}$ of
the Shortest Superstring problem and Maximal Compression problem with the following properties: 
\begin{enumerate}
\item[$(i)$] If there exists an assignment $\phi$ to the variables of $\cH$ which leaves at most  $u$ equations
unsatisfied, then, there exist a superstring $s_{\phi}$ for $\cS_{\cH}$  with  length
at most $|s_{\phi}|=5m_2+22 m_3+7n+u$.   
\item[$(ii)$] From every superstring $s$ for $S_{\cH}$  with length $|s|=5m_2+22 m_3+u+7n$,
we can construct in polynomial time  
an assignment $\psi_s$ to the variables
of $\cH$ that leaves at most $u$ equations in $\cH$ unsatisfied.
\item[$(iii)$] If there exists an assignment $\phi$ to the variables of $\cH$ which leaves at most  $u$ equations
unsatisfied, then, there exist a superstring $s_{\phi}$ for $\cS_{\cH}$  with  compression 
at least $comp(S_{\cH},s_{\phi})=3m_2+ 14 m_3-u+5n$ .
\item[$(iv)$] From every superstring $s$ for $S_{\cH}$  with  compression  
$comp(S_{\cH},s)=3m_2+ 14 m_3-u+5n$,
we can construct in polynomial time  
an assignment $\psi_s$ to the variables
of $\cH$ that leaves at most $u$ equations in $\cH$ unsatisfied.   
\item[$(v)$] The maximal orbit size of the instance $\cS_{\cH}$ is eight and the length 
 of a string in $\cS_{\cH}$ is bounded by six.
\end{enumerate}
\end{theorem}
\noindent
Combining Theorem~\ref{ssphybridsatz} with Theorem~\ref{hauptsatzsspI}, we obtain 
the following explicit lower bound for the Shortest Superstring problem.
%
%
\begin{corollary}
It is  \clNP-hard to approximate  the Shortest Superstring problem
with an approximation factor less than $345/344~(1.0029)$.
\end{corollary}
Before we proceed to the proof of Theorem~\ref{hauptsatzsspI}, we describe the 
reduction from a high-level view and try to build some intuition.

\subsection{Main Ideas and Overview}
Given an instance of the Hybrid problem $\cH$, we want to transform $\cH$ into an instance of the 
Shortest Superstring problem. Fortunately, the special structure of the linear 
equations in the Hybrid problem is particularly well-suited for our reduction,
since a part of the equations with two variables form a 
circle and every variable occurs exactly three times.
For every equation $g_{i+1}\equiv x_i\oplus x_{i+1}=0$ included in this circle, we introduce a set $S(g_{i+1})$ containing
two strings, which can be aligned advantageously in two natural ways. If those fragments 
corresponding to two successively following equations $x_{i-1}\oplus x_{i}=0$ and $x_i\oplus x_{i+1}=0$ use the same natural 
alignment, we are able to  overlap those fragments by
one letter. From a high level view, 
 we can construct  an associated superstring for each circle in $\cH$, which contains the natural aligned 
strings. 
In fact, we define for every equation  $g\in \cH$ an associated set
of strings $S(g)$ and the corresponding natural alignments. 
The instance $\cS_{\cH}$ of the Shortest Superstring problem is given by the union
of all sets $S(g)$.
Due to the construction of the sets $S(g)$, there is a particular
way to interpret an alignment of the strings in $S(g)$ included in  the resulting 
superstring as an assignment to the variables in the
Hybrid instance. 
The major challenge in the proof of
correctness is to prove that every superstring for $\cS_{\cH}$ can be interpreted as an
assignment to the variables in the Hybrid instance $\cH$ with the property that the number of
satisfied equations is connected to the length of the superstring.

\subsection{ Constructing $\mathcal{S}_{\cH} $ from $\cH$ }\label{sspconstr}
Given an instance of the Hybrid problem $\cH$, we are going to construct the corresponding
instance $\mathcal{S}_{\cH} $ of the Shortest Superstring problem. 
Furthermore, we introduce some notations and conventions.\\
\\
For every 
equation $g\in \cH$, we define a set $S(g)$ of corresponding strings.
The corresponding instance $\mathcal{S}_{\cH}$ of the Shortest Superstring problem is given by 
 $\mathcal{S}_{\cH} = \bigcup\limits_{g \in \mathcal{H} } S(g) $.
The strings in the set $S(g)$ differ by the type of considered equation $g\in \cH$. 
Let us  start 
with the description of $\cal{S}_{\cH} $. Therefore, we need to specify 
the instance of the Hybrid problem more precisely.\\ 
Let $\mathcal{E}_3$ be an instance of the MAX-E3-LIN problem and $\cH$ its corresponding  instance of the Hybrid problem  with $n$ circles.
For every variable $x\in V(\mathcal{E}_3)$, there is an associated circle $C^x$. Each circle consists of
 $m^x_2-1$ circle equations $g^x_{i+1}$ with $i\in [m^x_2-1]$, a circle border equation $g^x_1\equiv x_1\oplus x_{m^x_2}=0$ 
and $|M^x|$ matching equations $g^x_{e}$ with  $e\in M^x$. Furthermore, we have 
$m_3$ equations $g^3_j$ with exactly three variables. 
We are going to specify the sets $S(g)$ differing by the type of equation $g$.
In particular,  we  distinguish four types of equations contained in $\cH$.
\begin{enumerate}
\item[$(i)$] circle equations 
\item[$(ii)$] matching equations
\item[$(iii)$] circle border equations
\item[$(iv)$] equations with exactly three variables
\end{enumerate}
Let us begin with the description of the strings corresponding to circle border equations.
\subsubsection{Strings Corresponding to  Circle Border Equations } 
Given an instance of the hybrid problem $\cH$, a circle $C^x$ in $\cH$ and its
circle border equation $g^x_1\equiv x_1\oplus x_n=0$, we introduce six associated strings, that are all
included in
the set $S(g^x_1)$. Due to the construction of the circle $C^x$, the variable $x_n$ is a contact variable.
This means that $x_n$ appears in an equation $g^3_j$ with exactly three variables. The strings in the
set $S(g^x_1)$ differ by the type of equation $g^3_j$. We begin with the case 
$g^3_j\equiv x_n\oplus y\oplus z=0 $.\\
\\   
 The string $L_xC^l_x$ is used as the initial part of the superstring corresponding to this circle, whereas
 $ C^r_x R_x$ is used as the end part. Furthermore, we introduce 
 strings that represent an assignment that sets either  the  variable $x_1$ to $0$ or 
 the variable
$x_n$ to $1$.
 The corresponding two strings are
 $$
 C^l_x x^{m0}_1   x^{l1}_n  C^r_x  \qquad \textrm { and }
 \qquad  x^{l1}_n C^r_x   C^l_x x^{m0}_1.  
 $$
Finally, we define the last two strings of the set $S(g^x_1)$
$$C^l_x x^{r1}_1   x^{m0}_n  C^r_x \qquad 
 \textrm { and } \qquad x^{m0}_n C^r_x   C^l_x x^{r1}_1. 
$$
having a similar interpretation. Both pairs of strings can be overlapped by two letters. Those
natural alignments have a crucial influence during the process of constructing a superstring.
For this reason, we introduce a notation for this alignments. By the \emph{$0$-alignment} of the strings in $S(g^x_1)$,
we refer to the following alignment of the four strings.
In the following, $(\downarrow)$ will denote the overlapping of the strings.\\
 \vspace{1mm}\\
\fbox{\parbox{\dimexpr \linewidth - 2\fboxrule - 2\fboxsep}{
 $$
 C^l_x x^{m0}_1   x^{l1}_n  C^r_x  \quad \textrm { and }
 \quad  x^{l1}_n C^r_x   C^l_x x^{m0}_1 \qquad
 C^l_x x^{r1}_1   x^{m0}_n  C^r_x \quad 
 \textrm { and } \quad x^{m0}_n C^r_x   C^l_x x^{r1}_1. 
$$
$$
\downarrow \hspace{7cm} \downarrow 
$$
 $$C^l_x x^{m0}_1   x^{m1}_n  C^r_xC^l_x x^{m0}_1 
 \qquad \qquad \textrm { and } \qquad \qquad  x^{m0}_n C^r_xC^l_x x^{r1}_1   x^{m0}_n  C^r_x
 $$
 \vspace{1mm}
 }}
 \vspace{3mm}\\
On the other hand, we the define the \emph{$1$-alignment} of the strings in $S(g^x_1)$ as follows.\\
 \vspace{1mm}\\
\fbox{\parbox{\dimexpr \linewidth - 2\fboxrule - 2\fboxsep}{
 $$
 C^l_x x^{m0}_1   x^{l1}_n  C^r_x  \quad \textrm { and }
 \quad  x^{l1}_n C^r_x   C^l_x x^{m0}_1 \qquad
 C^l_x x^{r1}_1   x^{m0}_n  C^r_x \quad 
 \textrm { and } \quad x^{m0}_n C^r_x   C^l_x x^{r1}_1. 
$$
$$
\downarrow \hspace{7cm} \downarrow 
$$
 $$ x^{l1}_n C^r_x C^l_x x^{m0}_1   x^{l1}_n  C^r_x
 \qquad \qquad \textrm { and } \qquad \qquad  
 C^l_x x^{r1}_1   x^{m0}_n  C^r_x C^l_x x^{r1}_1
 $$
  \vspace{1mm}
 }}
 \vspace{3mm}\\
Both ways to join the four strings are called \emph{simple} alignments. \\
\\
After having described how the strings corresponding to $S(g^x_1)$ in case
of $g^3_j\equiv x_n\oplus y\oplus z=0 $ are defined, we are going to 
deal with the case $g^3_j\equiv x_n\oplus y\oplus z=1 $.\\   
As before, we use $L_xC^l_x$ as the initial part of the superstring corresponding to this circle, whereas
 $ C^r_x R_x$ is used as the end part. Furthermore, we define the remaining  four strings contained in 
 $S(g^x_1)$ by the following.\\
 \vspace{1mm}\\
\fbox{\parbox{\dimexpr \linewidth - 2\fboxrule - 2\fboxsep}{
$$
 C^l_x x^{m0}_1   x^{m1}_n  C^r_x  \qquad \qquad 
 \qquad  x^{m1}_n C^r_x   C^l_x x^{m0}_1  
 $$
$$
\textrm { and }
$$
$$C^l_x x^{r1}_1   x^{l0}_n  C^r_x \qquad \qquad 
 \qquad x^{l0}_n C^r_x   C^l_x x^{r1}_1 
$$
}}
 \vspace{3mm}\\
 Both pairs of strings can be overlapped by two letters. We introduce a notation for this
 alignments. \\  
 \vspace{1mm}\\
\fbox{\parbox{\dimexpr \linewidth - 2\fboxrule - 2\fboxsep}{
$$
 C^l_x x^{m0}_1   x^{m1}_n  C^r_x  \quad \textrm { and }
 \quad  x^{m1}_n C^r_x   C^l_x x^{m0}_1 \qquad
 C^l_x x^{r1}_1   x^{l0}_n  C^r_x \quad 
 \textrm { and } \quad x^{l0}_n C^r_x   C^l_x x^{r1}_1. 
$$
$$
\downarrow \hspace{7cm} \downarrow 
$$
 $$C^l_x x^{m0}_1   x^{m1}_n  C^r_xC^l_x x^{m0}_1 
 \qquad \qquad \textrm { and } \qquad \qquad  x^{l0}_n C^r_xC^l_x x^{r1}_1   x^{l0}_n  C^r_x
 $$
  \vspace{1mm}
}}
 \vspace{3mm}\\
 The former introduced  alignment is called the \emph{$0$-alignment} of the strings in $S(g^x_1)$.
 On the other hand, we the define the \emph{$1$-alignment} of the strings in $S(g^x_1)$ as follows.\\
\vspace{1mm}\\
\fbox{\parbox{\dimexpr \linewidth - 2\fboxrule - 2\fboxsep}{
 $$
 C^l_x x^{m0}_1   x^{m1}_n  C^r_x  \quad \textrm { and }
 \quad  x^{m1}_n C^r_x   C^l_x x^{m0}_1 \qquad
 C^l_x x^{r1}_1   x^{l0}_n  C^r_x \quad 
 \textrm { and } \quad x^{l0}_n C^r_x   C^l_x x^{r1}_1. 
$$
$$
\downarrow \hspace{7cm} \downarrow 
$$
 $$ x^{m1}_n C^r_x C^l_x x^{m0}_1   x^{m1}_n  C^r_x
 \qquad \qquad \textrm { and } \qquad \qquad  
 C^l_x x^{r1}_1   x^{l0}_n  C^r_x C^l_x x^{r1}_1
 $$
  \vspace{1mm}
 }}
 \vspace{3mm}\\
 In the remainder, we refer to both ways to overlap the four strings 
 as \emph{simple} alignments. 
Next, we describe the strings corresponding to
matching equations. 
\subsubsection{Strings Corresponding to Matching Equations}
Let $C^x$ be a circle in $\cH$ and $M^x$ its associated 
perfect matching. Let $\{i,j\}$ be an edge in $M^x$ and
$g^x_{\{i,j\}} \equiv x_i\oplus x_j=0$ the  associated matching equation.
We now define the corresponding set $S(g^x_{\{i,j\}})$ consisting
of two strings assuming $i<j$. Then, we introduce
two strings of the form 
$$x^{r0}_j x^{l0}_j x^{r1}_i x^{l1}_i \qquad 
\textrm{ and } 
\qquad
x^{r1}_i x^{l1}_i x^{r0}_j x^{l0}_j
$$
corresponding to the matching equation.
There are two ways to align those two strings to obtain an  overlap of two letters.
In the remainder, we refer to those alignments as \emph{simple}.\\
\vspace{1mm}\\
\fbox{\parbox{\dimexpr \linewidth - 2\fboxrule - 2\fboxsep}{
 \vspace{3mm}
$$
x^{r0}_j x^{l0}_j x^{r1}_i x^{l1}_i 
\quad
\textrm{ and } 
\quad
x^{r1}_i x^{l1}_i x^{r0}_j x^{l0}_j
$$
$$
\swarrow \qquad \searrow
$$
$$
\lefteqn{ \overbrace{ \phantom{x^{r0}_j x^{l0}_j x^{r1}_i x^{l1}_i    }  }
^{x^{r0}_j x^{l0}_j x^{r1}_i x^{l1}_i   }  }
x^{r0}_j x^{l0}_j 
\lefteqn{ \underbrace{ \phantom{x^{r1}_i x^{l1}_i x^{r0}_j x^{l0}_j    }  }
_{x^{r1}_i x^{l1}_i x^{r0}_j x^{l0}_j  }  }
x^{r1}_i x^{l1}_i x^{r0}_j x^{l0}_j
\qquad 
\qquad 
\lefteqn{ \overbrace{ \phantom{x^{r1}_i x^{l1}_i x^{r0}_j x^{l0}_j    }  }
^{x^{r1}_i x^{l1}_i x^{r0}_j x^{l0}_j  }  }
x^{r1}_i x^{l1}_i 
\lefteqn{ \underbrace{ \phantom{ x^{r0}_j x^{l0}_j x^{r1}_i x^{l1}_i   }  }
_{x^{r0}_j x^{l0}_j x^{r1}_i x^{l1}_i  }  }
x^{r0}_j x^{l0}_j x^{r1}_i x^{l1}_i
$$
 \vspace{1mm}
}}
\vspace{1mm}\\
The first way to overlap the strings is called the \emph{$0$-alignment}, whereas the second one
is called the \emph{$1$-alignment}. Next, we describe the strings corresponding to
circle equations. 
\subsubsection{Strings Corresponding to Circle  Equations}
Let $C^x$ be a circle in $\cH$ and $M^x$ its associated 
matching.
Furthermore, let  $\{i, j\}$ and
 $\{i+1, j'\}$ be both contained in $M^x$. We assume that $i<j$. Then, we 
 introduce the corresponding strings for $x_i \oplus x_{i+1}=0$.
 If $i+1<j'$, we have 
$$
x^{l1}_{i} x^{r1}_{i+1} x^{m0}_{i} x^{m0}_{i+1}
\qquad 
\textrm{ and }
\qquad 
x^{m0}_{i} x^{m0}_{i+1}  x^{l1}_{i} x^{r1}_{i+1}.
$$
Otherwise $(i+1>j')$, we have 
$$
x^{l1}_{i} x^{m1}_{i+1} x^{m0}_{i} x^{r0}_{i+1} 
\qquad \textrm{ and }
\qquad x^{m0}_{i} x^{r0}_{i+1}  x^{l1}_{i} x^{m1}_{i+1}.
$$
In case of $i>j$ and $i+1>j'$, we use 
$$
x^{m1}_{i} x^{m1}_{i+1} x^{l0}_{i} x^{r0}_{i+1} \qquad \textrm{ and }
\qquad x^{l0}_{i} x^{r0}_{i+1}  x^{m1}_{i} x^{m1}_{i+1}
$$
Finally, if $i>j$ and $i+1<j'$, we introduce
$$
x^{m1}_{i} x^{r1}_{i+1} x^{l0}_{i} x^{m0}_{i+1} \qquad \textrm{ and }
\qquad x^{l0}_{i} x^{m0}_{i+1}  x^{m1}_{i} x^{r1}_{i+1}
$$
Let $x_i$ be a variable in $\cH$ contained in an equation $g^3_j$ with
three variables. We now define the corresponding strings for the
equations $x_{i-1} \oplus x_{i}= 0$ and $x_i \oplus x_{i+1}= 0$. 
We assume that $\{i-1, j\}$ and $\{i+1, j'\}$ are both included in $M^x$.
Furthermore, we assume $i-1<j$ and $i+1<j'$. 
If the equation $g^3_j$ is of the form $x_i\oplus y \oplus z=0$, we introduce
$$
x^{l1}_{i-1} x^{r1}_{i} x^{m0}_{i-1} x^{m0}_{i} 
\qquad 
\textrm{ and }
\qquad 
x^{m0}_{i-1} x^{m0}_{i}  x^{l1}_{i-1} x^{r1}_{i}.
$$
for $x_{i-1} \oplus x_{i}= 0$. Furthermore, for $x_i \oplus x_{i+1}= 0$, we
use the strings 
$$
x^{l1}_{i} x^{r1}_{i+1} x^{m0}_{i} x^{m0}_{i+1} 
\qquad 
\textrm{ and }
\qquad 
x^{m0}_{i} x^{m0}_{i+1}  x^{l1}_{i} x^{r1}_{i+1}.
$$
On the other hand, if the equation $g^3_j$ is 
of the form $x_i\oplus y \oplus z=1$, we introduce
$$
x^{l1}_{i-1} x^{m1}_{i} x^{m0}_{i-1} x^{r0}_{i} 
\qquad
\textrm{ and }
\qquad 
x^{m0}_{i-1} x^{r0}_{i}  x^{l1}_{i-1} x^{m1}_{i}.
$$
corresponding to the equation $x_{i-1} \oplus x_{i}= 0$. For $x_i \oplus x_{i+1}= 0$, we
use the strings 
$$
x^{m1}_{i} x^{r1}_{i+1} x^{l0}_{i} x^{m0}_{i+1} 
\qquad 
\textrm{ and }
\qquad 
x^{l0}_{i} x^{m0}_{i+1}  x^{m1}_{i} x^{r1}_{i+1}.
$$
%
Accordingly, we introduce the notation of simple alignments for the strings in $S(g^x_{i+1})$.
For the strings 
$$
x^{m1}_{i} x^{m1}_{i+1} x^{r0}_{i} x^{l0}_{i+1} 
\qquad 
\textrm{ and }
\qquad 
x^{r0}_{i} x^{l0}_{i+1}  x^{m1}_{i} x^{m1}_{i+1},
$$
we define the following alignments as simple.
$$
\lefteqn{ \overbrace{ \phantom{x^{m1}_{i} x^{m1}_{i+1} x^{r0}_{i} x^{l0}_{i+1}     }  }
^{x^{m1}_{i} x^{m1}_{i+1} x^{r0}_{i} x^{l0}_{i+1}    }  }
x^{m1}_{i} x^{m1}_{i+1} 
\lefteqn{ \underbrace{ \phantom{x^{r0}_{i} x^{l0}_{i+1}  x^{m1}_{i} x^{m1}_{i+1}    }  }
_{x^{r0}_{i} x^{l0}_{i+1}  x^{m1}_{i} x^{m1}_{i+1}  }  }
x^{r0}_{i} x^{l0}_{i+1}  x^{m1}_{i} x^{m1}_{i+1}
\qquad
\textrm{ and } 
\qquad
\lefteqn{ \overbrace{ \phantom{x^{r0}_{i} x^{l0}_{i+1}  x^{m1}_{i} x^{m1}_{i+1}    }  }
^{x^{r0}_{i} x^{l0}_{i+1}  x^{m1}_{i} x^{m1}_{i+1}  }  }
x^{r0}_{i} x^{l0}_{i+1}  
\lefteqn{ \underbrace{ \phantom{ x^{m1}_{i} x^{m1}_{i+1} x^{r0}_{i} x^{l0}_{i+1}   }  }
_{x^{m1}_{i} x^{m1}_{i+1} x^{r0}_{i} x^{l0}_{i+1}  }  }
x^{m1}_{i} x^{m1}_{i+1} x^{r0}_{i} x^{l0}_{i+1}
$$ 
The the former alignment is called the $1$-alignment and the latter one is called the
$0$-alignment.  Next, we describe the strings corresponding to
 equations with three variables. 
\subsubsection{Strings Corresponding to Equations with Three Variables}
We now concentrate on equations with exactly three variables. 
Let $g^3_j$ be an equation with three variables in $\cH$. For every 
equation $g^3_j$, we define two  corresponding sets $S^A(g^3_j)$ and $S^B(g^3_j)$ 
 both containing exactly three strings. Finally, the set $S(g^3_j)$ is defined 
 by the union $S^A(g^3_j)\cup S^B(g^3_j)$. We distinguish whether $g^3_j$
is of the form  $x\oplus y \oplus z =1$ or $x\oplus y \oplus z =0$.
The description starts with the former case.\\
\\  
An equation of the form $ x\oplus y \oplus z =0$ is represented by $S^A(g^3_j)$
containing the strings
$$
x^{r1}A^1_jx^{l1}y^{r1}A^2_jy^{l1} \qquad y^{r1}A^2_jy^{l1} x^{m0} A^3_jC_j 
\qquad 
x^{m0} A^3_jC_jx^{r1} A^1_j x^{l1}
$$
and by $S^B(g^3_j)$ containing the strings
$$
x^{r1}B^1_jx^{l1}    z^{r1}B^2_jz^{l1} \qquad 
z^{r1}B^2_jz^{l1} C_j B^3_j x^{m0} \qquad 
C_jB^3_j x^{m0}  x^{r1} B^1_j  x^{l1}
$$
On the other hand, for equations of the form
 $g^3_j\equiv x\oplus y \oplus z =1$, we introduce $S^A(g^3_j)$ containing the following strings.
$$
x^{r0}A^1_jx^{l0}y^{r0}A^2_jy^{l0} 
\qquad 
y^{r0}A^2_jy^{l0} x^{m1} A^3_jC_j 
\qquad 
x^{m1} A^3_j C_jx^{r0} A^1_j x^{l0}
$$
Furthermore, we give the definition of $S^B(g^3_j)$  including the following strings.
$$
x^{r0}B^1_j x^{l0}    z^{r0}B^2_j z^{l0} 
\qquad z^{r0}B^2_j z^{l0} C_j B^3_jx^{m1} 
\qquad 
C_j  B^3_j x^{m1}  x^{r0} B^1_j  x^{l0}
$$
The strings in the set $S^A(g^3_j)$ can be aligned in a cyclic fashion in order to obtain
different strings which we will use in our reduction. Every specific alignment possesses
its own abbreviation given below.\\
\vspace{3mm}\\
\fbox{\parbox{\dimexpr \linewidth - 2\fboxrule - 2\fboxsep}{ 
 \vspace{1mm}
$$
x^{r1}A^1_jx^{l1}y^{r1}A^2_jy^{l1} \quad y^{r1}A^2_jy^{l1} x^{m0} A^3_jC_j 
\quad 
x^{m0} A^3_jC_jx^{r1} A^1_j x^{l1}
$$
$$
\downarrow \qquad \qquad
$$
$$
\lefteqn{ \overbrace{ \phantom{x^{r1}A^1_jx^{l1}y^{r1}A^2_jy^{l1}   }  }
^{x^{r1}A^1_jx^{l1}y^{r1}A^2_jy^{l1}  }  }
x^{r1}A^1_jx^{l1}
\lefteqn{ \underbrace{ \phantom{y^{r1}A^2_jy^{l1} x^{m0} A^3_jC_j   }  }
_{y^{r1}A^2_jy^{l1} x^{m0} A^3_jC_j  } }
y^{r1}A^2_jy^{l1} 
\lefteqn{ \overbrace{ \phantom{x^{m0} A^3_jC_jx^{r1} A^1_j x^{l1}   }  }
^{x^{m0} A^3_jC_jx^{r1} A^1_j x^{l1}  }  }
x^{m0} A^3_jC_jx^{r1} A^1_j x^{l1}
\equiv  x^{r1} A_j x^{l1} 
\textrm{ called $x^1$- alignment}
$$
$$
\downarrow
$$
$$
\lefteqn{ \underbrace{ \phantom{y^{r1}A^2_jy^{l1} x^{m0} A^3_jC_j   }  }
_{ y^{r1}A^2_jy^{l1} x^{m0} A^3_jC_j  } }
y^{r1}A^2_jy^{l1} 
\lefteqn{ \overbrace{ \phantom{x^{m0} A^3_jC_jx^{r1} A^1_j x^{l1}   }  }
^{x^{m0} A^3_jC_jx^{r1} A^1_j x^{l1}  }  }
x^{m0} A^3_jC_j 
\lefteqn{ \underbrace{ \phantom{x^{r1}A^1_jx^{l1}y^{r1}A^2_jy^{l1}  }  }
_{x^{r1}A^1_jx^{l1}y^{r1}A^2_jy^{l1}  }  }
x^{r1}A^1_jx^{l1}y^{r1}A^2_jy^{l1} \equiv  y^{r1} A_j y^{l1} 
\textrm{ called $y^1$- alignment}
$$
$$
\downarrow
$$
$$
\lefteqn{ \underbrace{ \phantom{x^{m0} A^3_jC_jx^{r1} A^1_j x^{l1}   }  }
_{x^{m0} A^3_jC_jx^{r1} A^1_j x^{l1}  }  }
x^{m0} A^3_jC_j
\lefteqn{ \overbrace{ \phantom{x^{r1}A^1_jx^{l1}y^{r1}A^2_jy^{l1}  }  }
^{x^{r1}A^1_jx^{l1}y^{r1}A^2_jy^{l1}  }  }
x^{r1}A^1_jx^{l1}
\lefteqn{ \underbrace{ \phantom{y^{r1}A^2_jy^{l1} x^{m0} A^3_jC_j  }  }
_{y^{r1}A^2_jy^{l1} x^{m0} A^3_jC_j } }
y^{r1}A^2_jy^{l1} x^{m0} A^3_jC_j \equiv x^{m0}A_jC_j
\textrm{ called left-$x^0$- alignment}
$$
}}
\vspace{3mm}\\
Analogously, the strings in $S^B(g^3_j)$ can also be aligned in a cyclic fashion.
We are going to define the abbreviations for these alignments.\\
\vspace{3mm}\\
\fbox{\parbox{\dimexpr \linewidth - 2\fboxrule - 2\fboxsep}{
 \vspace{1mm}
$$
x^{r1}B^1_jx^{l1}    z^{r1}B^2_jz^{l1} \quad 
z^{r1}B^2_jz^{l1} C_j B^3_j x^{m0} \quad 
C_jB^3_j x^{m0}  x^{r1} B^1_j  x^{l1}$$
$$
\downarrow
$$
$$
\lefteqn{ \underbrace{ \phantom{ C_jB^3_j x^{m0}  x^{r1} B^1_j  x^{l1}  }  }
_{C_jB^3_j x^{m0}  x^{r1} B^1_j  x^{l1 } } }
C_jB^3_j x^{m0}   
\lefteqn{ \overbrace{ \phantom{x^{r1}B^1_jx^{l1}    z^{r1}B^2_jz^{l1}   }  }
^{x^{r1}B^1_jx^{l1}    z^{r1}B^2_jz^{l1}  }  }
x^{r1}B^1_jx^{l1}   
\lefteqn{ \underbrace{ \phantom{z^{r1}B^2_jz^{l1} C_j B^3_j x^{m0}  }  }
_{z^{r1}B^2_jz^{l1} C_j B^3_j x^{m0} }  }
z^{r1}B^2_jz^{l1} C_j B^3_j x^{m0} \equiv C_jB_jx^{m0}
\textrm{ called right-$x^0$- alignment} 
$$
$$
\downarrow
$$
$$
\lefteqn{ \overbrace{ \phantom{x^{r1}B^1_jx^{l1}    z^{r1}B^2_jz^{l1}   }  }
^{x^{r1}B^1_jx^{l1}    z^{r1}B^2_jz^{l1} }  }
x^{r1}B^1_jx^{l1}   
\lefteqn{ \underbrace{ \phantom{z^{r1}B^2_jz^{l1} C_j B^3_j x^{m0}   }  }
_{z^{r1}B^2_jz^{l1} C_j B^3_j x^{m0}  }  }
z^{r1}B^2_jz^{l1} 
\lefteqn{ \overbrace{ \phantom{C_jB^3_j x^{m0}  x^{r1} B^1_j  x^{l1}    }  }
^{C_jB^3_j x^{m0}  x^{r1} B^1_j  x^{l1}   } }
C_jB^3_j x^{m0}  x^{r1} B^1_j  x^{l1}  \equiv x^{r1}B_jx^{l1}
\textrm{ called $x^1$- alignment}
$$
$$
\downarrow
$$
$$
\lefteqn{ \underbrace{ \phantom{z^{r1}B^2_jz^{l1} C_j B^3_j x^{m0}   }  }
_{z^{r1}B^2_jz^{l1} C_j B^3_j x^{m0}  }  }
z^{r1}B^2_jz^{l1}   
\lefteqn{ \overbrace{ \phantom{C_jB^3_j x^{m0}  x^{r1} B^1_j  x^{l1}  }  }
^{C_jB^3_j x^{m0}  x^{r1} B^1_j  x^{l1}  } }
C_jB^3_j x^{m0} 
\lefteqn{ \underbrace{ \phantom{x^{r1}B^1_jx^{l1}    z^{r1}B^2_jz^{l1}    }  }
_{x^{r1}B^1_jx^{l1}    z^{r1}B^2_jz^{l1}   }  }
x^{r1}B^1_jx^{l1}    z^{r1}B^2_jz^{l1}  \equiv z^{r1}B_j z^{l1}  
\textrm{ called $z^1$- alignment}
$$
}}
\vspace{3mm}\\
The strings in $S^B(g^3_j)$ and $S^A(g^3_j)$ can be overlapped
in a special way that corresponds to assigning the value $0$ to $x$.
\vspace{3mm}\\
\fbox{\parbox{\dimexpr \linewidth - 2\fboxrule - 2\fboxsep}{
 \vspace{3mm}
$$
x^{r1}A^1_jx^{l1}y^{r1}A^2_jy^{l1} \quad 
y^{r1}A^2_jy^{l1} x^{m0} A^3_jC_j \quad 
x^{m0} A^3_jC_jx^{r1} A^1_j x^{l1}
$$
$$
z^{r1}B^2_jz^{l1} C_j B^3_j x^{m0} \quad C_jB^3_j x^{m0}  x^{r1} B^1_j  x^{l1}
\quad 
x^{r1}B^1_jx^{l1}    z^{r1}B^2_jz^{l1}
$$
$$
\downarrow
$$
$$
\lefteqn{ \underbrace{ \phantom{x^{m0} A^3_jC_jx^{r1}A^1_jx^{l1}   }  }
_{x^{m0} A^3_jC_jx^{r1}A^1_jx^{l1}  }  }
x^{m0} A^3_jC_j
\lefteqn{ \overbrace{ \phantom{x^{r1}A^1_jx^{l1}y^{r1}A^2_jy^{l1}   }  }
^{x^{r1}A^1_jx^{l1}y^{r1}A^2_jy^{l1}  }  }
x^{r1}A^1_jx^{l1}
\lefteqn{ \underbrace{ \phantom{y^{r1}A^2_jy^{l1} x^{m0} A^3_jC_j  }  }
_{y^{r1}A^2_jy^{l1} x^{m0} A^3_jC_j  } }
y^{r1}A^2_jy^{l1} x^{m0} A^3_j
\lefteqn{ \overbrace{ \phantom{C_jB^3_j x^{m0}  x^{r1} B^1_j  x^{l1}  }  }
^{C_jB^3_j x^{m0}  x^{r1} B^1_j  x^{l1}  }  }
C_jB^3_j x^{m0}  
\lefteqn{ \underbrace{ \phantom{x^{r1}B^1_jx^{l1}    z^{r1}B^2_jz^{l1}   }  }
_{x^{r1}B^1_jx^{l1}    z^{r1}B^2_jz^{l1}  }  }
x^{r1}B^1_jx^{l1}      
\lefteqn{ \overbrace{ \phantom{z^{r1}B^2_jz^{l1} C_j B^3_j x^{m0}   }  }
^{z^{r1}B^2_jz^{l1} C_j B^3_j x^{m0}  } }
z^{r1}B^2_jz^{l1} C_j B^3_j x^{m0}
$$
 \vspace{3mm}
}}
\vspace{1mm}\\
In the remainder, we call this alignment the $x^0$-alignment of $S(g^3_j)$ and use the
abbreviation $x^{m0}C_jx^{m0}$ for this string.
\subsection{   Constructing the Superstring $s_{\phi}$ from $\phi$}\label{sspgivenphi}
Given an assignment $\phi$ to the variables of $\cH$, we are going to 
construct  
the associated superstring $s_{\phi}$ for 
the instance $\cal{S}_{\cH} $.\\
\\
For every $g\in \cH$, we  formulate rules for
aligning the corresponding strings in $S(g)$ according to the 
assignment $\phi$.  We start with sets corresponding to
circle border equations and circle equations. 
Afterwards, we show how the actual fragments
can be overlapped with strings from the sets corresponding
to matching equations and equations with three variables. 
Furthermore, we analyze the relation between the assignment 
$\phi$ and the length of the obtained superstring $s_{\phi}$.
We begin with the description of the alignment of 
strings corresponding to circle border equations in $\cH$.

\subsubsection{Aligning Strings Corresponding to Circle Border Equations }
Let $C^x$ be 
 a circle  in $\cH$ and $x_1\oplus x_n=0$ its circle  border equation. Furthermore, we assume
 that $x_n$ is contained in a equation with three variables of the form $x_n\oplus y\oplus z=0$.
 First, we set
 the string $L_xC^l_x$  as the initial part of our superstring corresponding to the circle $C^x$.
 Then, we use  the $\phi(x_1)$-alignment of the strings
$$C^l_x x^{m0}_1  x^{m1}_n  C^r_x,   \quad
  x^{l1}_n C^r_x   C^l_x x^{m0}_1,
\quad 
C^l_x x^{r1}_1   x^{m0}_n  C^r_x ,  \quad
\textrm {and } \quad  x^{m0}_n C^r_x   C^l_x x^{r1}_1.   $$ 
 In this condition, one of the strings $s_l$ can be overlapped  from the left side 
 with $L_xC^l_x$ by one letter. The other string $s_r$ will be joined from the right side 
 with $ C^r_x R_x$
by one letter. 
This construction will help us to check whether $\phi$   assigns the same
value to the variable $x_n$ as to $x_1$. The string $s_r$ can be interpreted as the $\phi(x_1)$-alignment of the 
strings corresponding to $x_n\oplus x_{n+1}=0$, since the first letter of $s_r$ is either
$x^{m0}_n$ or $x^{l1}_n$. \\ 
\\
The parts corresponding to a circle border equation with $x_n\oplus y\oplus z =1$
can be constructed analogously. Next, we are going to align strings corresponding
to circle equations.
\subsubsection{Aligning Strings Corresponding to Circle Equations }
Let $x_i\oplus x_{i+1}=0$ be a circle equation contained in $\cH$.
Furthermore, let the corresponding strings are given by
$$
 x^{m0}_{i} x^{m0}_{i+1} x^{l1}_{i} x^{r1}_{i+1} \quad \textrm{ and }
 \quad 
 x^{l1}_{i} x^{r1}_{i+1} x^{m0}_{i} x^{m0}_{i+1}.
$$
In dependence of the given assignment $\phi$, we use
simple alignments to overlap the considered strings.
More precisely,  we make use of the $\phi(x_{i+1})$-alignment.
For every pair of associated strings, we derive an overlap of two letters.
We are going to align those fragments with strings corresponding to
 matching equations and equations with three variables.
\subsubsection{Aligning Strings Corresponding to Matching Equations}
Let $x_i \oplus x_j=0$ be a matching equation in $\cH$.
Let us assume that $i<j$. We define the alignment 
of the strings in $S(g^x_{\{i,j\}})$ according to the value of
$\phi (x_{i+1})$. More precisely, we use the $\phi (x_{i+1})$-alignment
of the strings
$$
x^{r0}_j x^{l0}_j x^{r1}_i x^{l1}_i 
\quad
\textrm{ and }
\quad
x^{r1}_i x^{l1}_i x^{r0}_j x^{l0}_j. 
$$
Due to this alignment, we obtain an overlap of two letters.
We are going to analyze the length of the resulting  superstring
in dependence of the assignment $\phi$ to the variables 
$x_{i}$, $x_{i+1}$, $x_{j}$ and $x_{j+1}$.
We start with the case  $\phi(x_{i+1})=\phi(x_{j+1})=1$.\\
\\ 
\textbf{Case $\phi(x_{i+1})=\phi(x_{j+1})=1$:}\\
We use the $1$-alignment of the strings $x^{r1}_i x^{l1}_i x^{r0}_j x^{l0}_j$
and $x^{r0}_j x^{l0}_j x^{r1}_i x^{l1}_i$. The situation  is depicted below.
(The two triangle notation $\medtriangleright\medtriangleright$ and $\medtriangleleft\medtriangleleft$ 
will be explained hereafter.)
\\
\vspace{4mm}\\
\fbox{\parbox{\dimexpr \linewidth - 2\fboxrule - 2\fboxsep}{
\vspace{3mm}
$$
b
 \medtriangleright\medtriangleright
  \framebox{$X_i$} \,
x^{r0}_j x^{l0}_j x^{r1}_i x^{l1}_i \, 
\framebox{$x^{l1}_i$} 
 \medtriangleleft\medtriangleleft m
\medtriangleright\medtriangleright
 \framebox{$Y_j$} \,
x^{r1}_i x^{l1}_i x^{r0}_j x^{l0}_j \,
\framebox{$x^{m1}_j$}
\medtriangleleft\medtriangleleft
e
$$
$$\downarrow$$
$$
b
\medtriangleright\medtriangleright
 \framebox{$X_i$} \,
\lefteqn{ \overbrace{ \phantom{x^{r1}_i x^{l1}_i x^{r0}_j x^{l0}_j   }  }
^{x^{r1}_i x^{l1}_i x^{r0}_j x^{l0}_j  }         }
x^{r1}_i x^{l1}_i 
\lefteqn{ \underbrace{ \phantom{x^{r0}_j x^{l0}_j x^{r1}_i \framebox{$x^{l1}_i$}   } } 
_{  x^{r0}_j x^{l0}_j x^{r1}_i x^{l1}_i    }         }
x^{r0}_j x^{l0}_j x^{r1}_i 
 \framebox{$x^{l1}_i$}
 \medtriangleleft\medtriangleleft
 m
 \medtriangleright\medtriangleright
 \framebox{$Y_j^{\,}$}\,
\framebox{$x^{m1}_j$}
\medtriangleleft\medtriangleleft
e
$$
\vspace{1mm}
}}
\vspace{1mm}\\
The actual superstring $s$ is denoted by the following sequence.
$$s=b
 \medtriangleright\medtriangleright
  \framebox{$X_i$} \,
x^{r0}_j x^{l0}_j x^{r1}_i x^{l1}_i \, 
\framebox{$x^{l1}_i$} 
 \medtriangleleft\medtriangleleft m
\medtriangleright\medtriangleright
 \framebox{$Y_j$} \,
x^{r1}_i x^{l1}_i x^{r0}_j x^{l0}_j \,
\framebox{$x^{m1}_j$}
\medtriangleleft\medtriangleleft
e$$
The part $\medtriangleright\medtriangleright \framebox{$X_i$}$ represents 
a simple alignment of the 
strings corresponding
to $x_{i-1}\oplus x_i=0$
 ending with the letter
$X_i\in \{x^{m0}_i, x^{r1}_i    \}$,
which means  
$$
\medtriangleright\medtriangleright \framebox{$X_i$}\in 
\{
x^{m0}_{i-1} x^{m0}_{i}  x^{l1}_{i-1} x^{r1}_{i} x^{m0}_{i-1} x^{m0}_{i},
x^{l1}_{i-1} x^{r1}_{i}  x^{m0}_{i-1} x^{m0}_{i}  x^{l1}_{i-1} x^{r1}_{i}
\}.
$$ 
The letter in the box emphasizes the letter which can be used to
 overlap from the right side with other strings. 
Furthermore, the string $\framebox{$x^{l1}_i$} \medtriangleleft\medtriangleleft $ denotes
$x^{l1}_{i} x^{r1}_{i+1}  x^{m0}_{i} x^{m0}_{i+1}  x^{l1}_{i} x^{r1}_{i+1}$. 
Analogously,  
$\medtriangleright\medtriangleright \framebox{$Y_j$} $ is a simple alignment of  the strings corresponding to
 $x_{j-1}\oplus x_j=0$, where
$Y_j\in \{  x^{r0}_j, x^{m1}_j \}$. Furthermore, we use
$\framebox{$x^{m1}_j$} \medtriangleleft\medtriangleleft$ to denote 
$x^{m1}_j   x^{m1}_{j+1}  x^{l0}_{j} x^{r0}_{j+1}  x^{m1}_{j} x^{m1}_{j+1} $.
Finally, $b$, $m$ and $e$ are sequences of letters, which we do not specify in detail.
They define the remaining parts of the superstring $s$. 
\\
If $X_i=x^{r1}_i$ holds, we align $\medtriangleright\medtriangleright \framebox{$X_i$}$ 
with $x^{r1}_i x^{l1}_i x^{r0}_j x^{l0}_j  x^{r1}_i x^{l1}_i$
to achieve an additional overlap of one letter.
An analogue situation holds for $\medtriangleright\medtriangleright \framebox{$Y_j$}$ and 
$\framebox{$x^{m1}_j$} \medtriangleleft\medtriangleleft$. 
All in all, we obtain an overlap of three letters if $\phi(x_{i})=\phi(x_{i+1})=1$
and $\phi(x_{j+1})=\phi(x_j)=1$ holds. Otherwise, we lose an overlap of one letter 
per unsatisfied equation. \\
\\ 
\textbf{Case $\phi(x_{i+1})=\phi(x_{j+1})=0$:}\\
We use the $0$-alignment of the strings $x^{r1}_i x^{l1}_i x^{r0}_j x^{l0}_j$
and $x^{r0}_j x^{l0}_j x^{r1}_i x^{l1}_i$.\\
\vspace{3mm}\\
\fbox{\parbox{\dimexpr \linewidth - 2\fboxrule - 2\fboxsep}{
\vspace{4mm}
$$
b
\medtriangleright\medtriangleright \framebox{$X_i$} \,
x^{r0}_j x^{l0}_j x^{r1}_i x^{l1}_i \, 
\framebox{$x^{m0}_i$} \medtriangleleft\medtriangleleft \,m
\medtriangleright\medtriangleright \framebox{$Y_j$} \,
x^{r1}_i x^{l1}_i x^{r0}_j x^{l0}_j \,
\framebox{$x^{l0}_j$} \medtriangleleft\medtriangleleft \,
e
$$
$$\downarrow$$
$$
b \,
\medtriangleright\medtriangleright \framebox{$X_i$} \,
\framebox{$x^{m0}_i$} \medtriangleleft\medtriangleleft \,m
\medtriangleright\medtriangleright \framebox{$Y_j$} \,
\lefteqn{ \overbrace{ \phantom{x^{r0}_j x^{l0}_j x^{r1}_i x^{l1}_i   } } 
^{  x^{r0}_j x^{l0}_j x^{r1}_i x^{l1}_i    }         }
x^{r0}_j x^{l0}_j
\lefteqn{ \underbrace{ \phantom{x^{r1}_i x^{l1}_i x^{r0}_j \framebox{$x^{l0}_j$}   }  }
_{x^{r1}_i x^{l1}_i x^{r0}_j x^{l0}_j  }         }
x^{r1}_i x^{l1}_i x^{r0}_j 
\framebox{$x^{l0}_j$}\medtriangleleft\medtriangleleft
\,e
$$\vspace{1mm}
}}
\vspace{1mm}\\
In this case, we use $\framebox{$x^{m0}_i$} \medtriangleleft\medtriangleleft$   as an abbreviation for 
$  x^{m0}_{i} x^{m0}_{i+1}  x^{l1}_{i} x^{r1}_{i+1}  x^{m0}_{i} x^{m0}_{i+1}$
and $\framebox{$x^{l0}_j$}\medtriangleleft\medtriangleleft$ for 
$   x^{l0}_{j} x^{r0}_{j+1}  x^{m1}_{j} x^{m1}_{j+1} x^{l0}_{j} x^{r0}_{j+1}$.
If $X_i=x^{m0}_i$ holds, we align $\medtriangleright\medtriangleright \framebox{$X_i$}$ 
with $\framebox{$x^{m0}_i$} \medtriangleleft\medtriangleleft$ and gain an additional overlap of
 one letter.
An analogue situation holds for $\medtriangleright\medtriangleright \framebox{$Y_j$}$ and 
$\framebox{$x^{l0}_j$} \medtriangleleft\medtriangleleft$. 
Hence, we obtain an overlap of three letters if $\phi(x_{i+1})=\phi(x_i)=0$
and $\phi(x_{j+1})=\phi(x_j)=0$ holds. If the corresponding equation with two 
variables is not satisfied, we lose an overlap of one letter.
 \\
\\
\textbf{Case $\phi(x_{i+1})\neq \phi(x_{j+1})=1$:}\\
In this case, we use the $0$-alignment of the strings $x^{r1}_i x^{l1}_i x^{r0}_j x^{l0}_j$
and $x^{r0}_j x^{l0}_j x^{r1}_i x^{l1}_i$.\\
\vspace{3mm}\\
\fbox{\parbox{\dimexpr \linewidth - 2\fboxrule - 2\fboxsep}{ 
\vspace{4mm}
$$
b
\medtriangleright\medtriangleright \framebox{$X_i$} \,
x^{r0}_j x^{l0}_j x^{r1}_i x^{l1}_i \, 
\framebox{$x^{m0}_i$} \medtriangleleft\medtriangleleft \,
m
\medtriangleright\medtriangleright \framebox{$Y_j$} \,
x^{r1}_i x^{l1}_i x^{r0}_j x^{l0}_j \,
\framebox{$x^{m1}_j$} \medtriangleleft\medtriangleleft
e
$$
$$\downarrow$$
$$
b
\medtriangleright\medtriangleright \framebox{$X_i$} \,
\framebox{$x^{m0}_i$} \medtriangleleft\medtriangleleft \,
m
\medtriangleright\medtriangleright \framebox{$Y_j$} \,
\framebox{$x^{m1}_j$} \medtriangleleft\medtriangleleft e \quad
\lefteqn{ \overbrace{ \phantom{ x^{r0}_j x^{l0}_j x^{r1}_i x^{l1}_i  }  }
^{x^{r0}_j x^{l0}_j x^{r1}_i x^{l1}_i  }         }
x^{r0}_j x^{l0}_j 
\lefteqn{ \underbrace{ \phantom{ x^{r1}_i x^{l1}_i x^{r0}_j x^{l0}_j  } } 
_{  x^{r1}_i x^{l1}_i x^{r0}_j x^{l0}_j    }         }
x^{r1}_i x^{l1}_i x^{r0}_j x^{l0}_j
$$
\vspace{1mm}
}}
\vspace{1mm}\\
We attach  $x^{r0}_j x^{l0}_j x^{r1}_i x^{l1}_i x^{r0}_j x^{l0}_j$
 at the end of our actual solution $s_{\phi}$ without having any overlap with the so far obtained superstring. 
Notice that we obtain in each case an additional overlap of one letter  if the corresponding 
equation with two variables  is satisfied, i.e. $X_i=x^{m0}_i$ and $Y_j=x^{m1}_j$.\\
\\
\textbf{Case $\phi(x_{i+1})\neq \phi(x_{j+1})=0$:}\\
According to $\phi$, we use the $1$-alignment of the strings $x^{r1}_i x^{l1}_i x^{r0}_j x^{l0}_j$
and $x^{r0}_j x^{l0}_j x^{r1}_i x^{l1}_i$.\\
\vspace{3mm}\\
\fbox{\parbox{\dimexpr \linewidth - 2\fboxrule - 2\fboxsep}{
\vspace{4mm}
$$
b
 \medtriangleright\medtriangleright
  \framebox{$X_i$} \,
x^{r0}_j x^{l0}_j x^{r1}_i x^{l1}_i \, 
\framebox{$x^{l1}_i$} 
 \medtriangleleft\medtriangleleft m
\medtriangleright\medtriangleright
 \framebox{$Y_j$} \,
x^{r1}_i x^{l1}_i x^{r0}_j x^{l0}_j \,
\framebox{$x^{l0}_j$}
\medtriangleleft\medtriangleleft
e
$$
$$\downarrow$$
$$
b
\medtriangleright\medtriangleright
 \framebox{$X_i$} \,
\lefteqn{ \overbrace{ \phantom{x^{r1}_i x^{l1}_i x^{r0}_j x^{l0}_j   }  }
^{x^{r1}_i x^{l1}_i x^{r0}_j x^{l0}_j  }         }
x^{r1}_i x^{l1}_i 
\lefteqn{ \underbrace{ \phantom{x^{r0}_j x^{l0}_j x^{r1}_i \framebox{$x^{l1}_i$}   } } 
_{  x^{r0}_j x^{l0}_j x^{r1}_i x^{l1}_i    }         }
x^{r0}_j x^{l0}_j x^{r1}_i 
 \framebox{$x^{l1}_i$}
 \medtriangleleft\medtriangleleft
 m
 \medtriangleright\medtriangleright
 \framebox{$Y_j^{\,}$}\,
\framebox{$x^{l0}_j$}
\medtriangleleft\medtriangleleft
e
$$
\vspace{1mm}
}}
\vspace{1mm}\\
We join  $ x^{r1}_i x^{l1}_i x^{r0}_j x^{l0}_j x^{r1}_i x^{l1}_i$ 
from the right side
with
$\framebox{$x^{l1}_i$} 
 \medtriangleleft\medtriangleleft$  and obtain an overlap 
 of one letter. This reduces the length of the superstring by one letter
 independent of the assignment $\phi(x_j)$. In case of $X_i=x^{r1}_i$, we
 achieve another overlap of one letter, since we are able to align 
 $\medtriangleright\medtriangleright
 \framebox{$X_i$}$ from the right side with 
$ x^{r1}_i x^{l1}_i x^{r0}_j x^{l0}_j x^{r1}_i x^{l1}_i$.
It corresponds to the satisfied equation $x_i\oplus x_{i+1}=0$. 
Hence, we obtain at least the same number of overlapped letters as
satisfied equations.\\
\\
In summary, we note that we are able to achieve 
an  overlap of at least 
one letter in each case if the corresponding equation is satisfied by $\phi$.
Hence, we obtain an overlap of at most three letters. \\
\\
The other cases concerning equations $x_i\oplus x_j=0$ with $i>j$
can be analyzed analogously.
Next, we are going to align strings corresponding to equations with three variables.

\subsubsection{ Aligning Strings Corresponding to Equations with Three Variables}
Let $g^3_j\in \cH$ be an equation with three variables $x$, $y$ and $z$.
Furthermore, let $x_{i-1}\oplus x =0$, $x \oplus x_{i+1} =0$, $y_{j-1}\oplus y =0$,
$y \oplus y_{j+1} =0$,  $z_{k-1}\oplus z =0$ and $z \oplus z_{k+1} =0$
be the equations with two variables, in which the variables $x$, $y$ and $z$ occur.
Given the assignment $\phi$ to $x$, $y$ and $z$, we are going to define the alignment of  
 the corresponding strings.
Let us start with equations of the form  $g^3_j\equiv x\oplus y\oplus z= 0 $.
Then, we define the 
rule for aligning strings in $S^A(g^3_j)$ and $S^B(g^3_j)$ as follows, whereby
we handle the cases $\phi(x_{i+1})+\phi(y_{j+1})+\phi(z_{k+1})=\{3,2,1,0\}$ separately starting
with $\phi(x_{i+1})+\phi(y_{j+1})+\phi(z_{k+1})=3$.
\\
\\
%
\textbf{Case $\phi(x_{i+1})+\phi(y_{j+1})+\phi(z_{k+1})=3$:} \\
In this case, we align the strings in $S(g^3_j)$ in such a way that we obtain the
former introduced strings $y^{r1}A_jy^{l1}$ and $z^{r1}B_jz^{l1}$.
The situation, which we want to analyze, is depicted below.\\
\vspace{3mm}\\
\fbox{\parbox{\dimexpr \linewidth - 2\fboxrule - 2\fboxsep}{
$$
\vspace{2mm}
b\,
\medtriangleright\medtriangleright \framebox{$X$} \,
 \framebox{$x^{l1}$} 
 \medtriangleleft\medtriangleleft \, m_1
\medtriangleright\medtriangleright \framebox{$Y$}\
 y^{r1} A_j y^{l1}\
  \framebox{$y^{l1}$} 
  \medtriangleleft\medtriangleleft \, m_2
\medtriangleright\medtriangleright \framebox{$Z$} \
z^{r1}B_jz^{l1}\
 \framebox{$z^{l1}$} 
  \medtriangleleft\medtriangleleft
  \, e
$$
$$
\downarrow
$$
$$
b\,
\medtriangleright\medtriangleright \framebox{$X$} \,
 \framebox{$x^{l1}$} 
 \medtriangleleft\medtriangleleft \, m_1
\medtriangleright\medtriangleright \framebox{$Y$}\
\lefteqn{ \underbrace{ \phantom{y^{r1} A_j \framebox{$y^{l1}$}  }  }
_{y^{r1} A_j y^{l1}  }  } 
 y^{r1} A_j 
 \framebox{$y^{l1}$} \medtriangleleft\medtriangleleft \, m_2
\medtriangleright\medtriangleright \framebox{$Z$} \
 \lefteqn{ \overbrace{ \phantom{ z^{r1}B_j\framebox{$z^{l1}$}    }  }
^{z^{r1}B_jz^{l1}   }  } 
z^{r1}B_j
\framebox{$z^{l1}$} \medtriangleleft\medtriangleleft 
\, e
$$
\vspace{1mm}
}}
\vspace{1mm}\\
Similarly to the situations that we discussed concerning matching equations, we
define the actual superstring $s$ in the way described below.
$$
s=b\,
\medtriangleright\medtriangleright \framebox{$X$} \,
 \framebox{$x^{l1}$} 
 \medtriangleleft\medtriangleleft \, m_1
\medtriangleright\medtriangleright \framebox{$Y$}\
 y^{r1} A_j y^{l1}\
  \framebox{$y^{l1}$} 
  \medtriangleleft\medtriangleleft \, m_2
\medtriangleright\medtriangleright \framebox{$Z$} \
z^{r1}B_jz^{l1}\
 \framebox{$z^{l1}$} 
  \medtriangleleft\medtriangleleft
  \, e
$$
Here, $b$, $m_1$, $m_2$ and $e$ denote parts of $s$, which we do not specify
in detail to emphasize the parts corresponding to the equation with three variables. \\
The string $\framebox{$x^{l1}$} 
 \medtriangleleft\medtriangleleft$ denotes the $\phi(x_{i+1})$-alignment
 of the corresponding strings in $S(g^x_{i+1})$.
The strings $\framebox{$z^{l1}$} 
  \medtriangleleft\medtriangleleft$ and $\framebox{$y^{l1}$} 
  \medtriangleleft\medtriangleleft$ are defined analogously.
  In this situation, we want to analyze the cases 
$X\in \{x^{r1}, x^{m0}\}$, $Y\in \{y^{r1}, y^{m0}\}$
and $Z\in \{z^{r1}, z^{m0}\}$. 
We infer that we obtain an overlap of four letters if all equations with
two variables are satisfied. Otherwise, we lose an overlap of one letter
per  unsatisfied equation with two variables.\\ 
\\
\textbf{Case $\phi(x_{i+1})+\phi(y_{j+1})+\phi(z_{k+1})= 2$:} \\
Let $\alpha,\gamma \in \{x_{i+1},y_{j+1},z_{k+1}\}$ be  variables  such that $\phi(\gamma)=\phi(\alpha)=1$ holds.
Then, we use the $\alpha^1$-alignment and $\gamma^1$-alignment of the strings in $S^A(g^3_j)$ and  
 $S^B(g^3_j)$ breaking ties arbitrary.
We display exemplary the situation for $\phi(z_{k+1})=\phi(x_{i+1})=1$.\\
\vspace{3mm}\\
\fbox{\parbox{\dimexpr \linewidth - 2\fboxrule - 2\fboxsep}{
\vspace{2mm}
$$
b
\medtriangleright\medtriangleright \framebox{$X$} \, 
x^{r1} A_j x^{l1}\, 
\framebox{$x^{l1}$} \medtriangleleft\medtriangleleft m_1
\medtriangleright\medtriangleright \framebox{$Y$}\,
\framebox{$y^{m0}$} \medtriangleleft\medtriangleleft m_2
\medtriangleright\medtriangleright \framebox{$Z$} \,
z^{r1}B_jz^{l1}\,
 \framebox{$z^{l1}$} \medtriangleleft\medtriangleleft e
$$
$$
\downarrow
$$
$$
b
\medtriangleright\medtriangleright \framebox{$X$}\
\lefteqn{ \underbrace{ \phantom{x^{r1} A_j \framebox{$x^{l1}$}  }  }
_{x^{r1} A_j x^{l1}  }  } 
 x^{r1} A_j 
 \framebox{$x^{l1}$} \medtriangleleft\medtriangleleft
 m_1
  \medtriangleright\medtriangleright \framebox{$Y$}\
\framebox{$y^{m0}$} \medtriangleleft\medtriangleleft
m_2
\medtriangleright\medtriangleright \framebox{$Z$} \
 \lefteqn{ \overbrace{ \phantom{ z^{r1}B_j\framebox{$z^{l1}$}    }  }
^{z^{r1}B_jz^{l1}   }  } 
z^{r1}B_j
\framebox{$z^{l1}$} \medtriangleleft\medtriangleleft
e
$$
\vspace{2mm}
}}
\vspace{1mm}\\
In this case, we achieve an overlap of five letters if all equations with
two variables are satisfied. Otherwise, we lose an overlap of one letter
per unsatisfied equation with two variables.\\
\\
\textbf{Case $\phi(x_{i+1})+\phi(y_{j+1})+\phi(z_{k+1})= 1$:} \\
If $\phi(z_{k+1})+\phi(x_{i+1})=1$ holds, we align the strings in $S^B(g^3_j)$ and $S^A(g^3_j)$
to obtain $x^{r1} A_j x^{l1}$ and $z^{r1} B_j z^{l1}$. Otherwise, we 
make use of the strings $x^{r1} B_j x^{l1}$ and $y^{r1} A_j y^{l1}$.
We display the  situation for $\phi(y_{j+1})=1$.\\
\vspace{3mm}\\
\fbox{\parbox{\dimexpr \linewidth - 2\fboxrule - 2\fboxsep}{
\vspace{2mm}
$$
b
\medtriangleright\medtriangleright \framebox{$X$}  \,
x^{r1} B_j x^{l1}
\framebox{$x^{m0}$}\medtriangleleft\medtriangleleft m_1
\medtriangleright\medtriangleright \framebox{$Y$} \,
 y^{r1} A_j y^{l1} \framebox{$y^{l1}$} \medtriangleleft\medtriangleleft 
 m_2 
\medtriangleright\medtriangleright \framebox{$Z$} \,
 \framebox{$z^{m0}$}\medtriangleleft\medtriangleleft  e
$$
$$
\downarrow
$$
$$
b
\medtriangleright\medtriangleright
 \framebox{$X$} \, \framebox{$x^{m0}$} \medtriangleleft\medtriangleleft  
 m_1
\medtriangleright\medtriangleright \framebox{$Y$} \,
\lefteqn{ \overbrace{ \phantom{y^{r1} A_j \framebox{$y^{l1}$}   }  }
^{y^{r1} A_j y^{l1}  }  }
y^{r1}  A_j
\framebox{$y^{l1}$} \medtriangleleft\medtriangleleft 
m_2
\medtriangleright\medtriangleright
 \framebox{$Z$} \, \framebox{$z^{m0}$} \medtriangleleft\medtriangleleft
 e\, 
x^{r1} B_j x^{l1}
$$\vspace{2mm}
}}
\vspace{1mm}\\
Notice that  we obtain an overlap of four letters if the equations with two variables 
are satisfied, i.e. $X=x^{m0}$, $Z=z^{m0}$ and $Y=y^{r1}$.
Otherwise, we lose an overlap of one letter per unsatisfied equation with two variables.\\
\\
\textbf{Case $\phi(x_{i+1})+\phi(y_{j+1})+\phi(z_{k+1})= 0$:} \\
In this case, we use the $x^0$-alignment of the strings in $S(g^3_j)$.
The situation is displayed below.\\
\vspace{3mm}\\
\fbox{\parbox{\dimexpr \linewidth - 2\fboxrule - 2\fboxsep}{
\vspace{4mm}
$$
b
\medtriangleright\medtriangleright \framebox{$X$} \,
x^{m0}  
C_j  x^{m0}\
\framebox{$x^{m0}$} \medtriangleleft\medtriangleleft 
m_1
\medtriangleright\medtriangleright \framebox{$Y$} \,
 \framebox{$y^{m0}$} \medtriangleleft\medtriangleleft  
 m_2
\medtriangleright\medtriangleright \framebox{$Z$} \,
\framebox{$z^{m0}$} \medtriangleleft\medtriangleleft
e 
$$
$$
\downarrow
$$
$$
b
\medtriangleright\medtriangleright \framebox{$X$} \
 \lefteqn{ \underbrace{ \phantom{x^{m0} C_j  \framebox{$x^{m0}$}   }  }
_{x^{m0} C_j  x^{m0}   }  }
x^{m0} C_j 
\framebox{$x^{m0}$} \medtriangleleft\medtriangleleft
m_1
\medtriangleright\medtriangleright \framebox{$Y$} \,
 \framebox{$y^{m0}$}\medtriangleleft\medtriangleleft
 m_2
\medtriangleright\medtriangleright\framebox{$Z$} \, 
\framebox{$z^{m0}$} \medtriangleleft\medtriangleleft
e
$$
\vspace{1mm}
}}
\vspace{1mm}\\
Here, we are able to achieve an 
 overlap of five letters  if all equations with two variables are satisfied,
 i.e. $X=x^{m0}$, $Z=z^{m0}$ and $Y=y^{m0}$.\\
 \\
In summary, we state that we can achieve 
an overlap of at least one letter independent of the assignment $\phi$. 
Additionally, we gain another overlap of 
one letter if the corresponding equation is satisfied by $\phi$.\\
\\
The situation for equations of the form $x\oplus y\oplus z= 1$
can be analyzed analogously. We are going to define the assignment $\psi_s$,
which is associated to a given superstring for $\cS_{\cH}$.

%
\subsection{ Defining the Assignment }\label{sspassignment}
Given a superstring $s$ for $\cS_{\cH}$, we are going to define the associated
assignment $\psi_s$ to the variables of $\cH$. In order to deduce 
the values assigned to the variables in $\cH$ from $s$,
 we have to normalize 
the given superstring $s$. For this reason, we define rules
that transform a superstring for $\cS_{\cH}$ into a normed
superstring for $\cS_{\cH}$ without increasing the length.\\
\\
First, we introduce the definition of a normed superstring for $\cS_{\cH}$.
\begin{definition}[Normed Superstring $s$ for $\cS_{\cH}$]
Let $\cH$ be an instance of the Hybrid problem, $\cS_{\cH}$
the corresponding instance of the Shortest Superstring problem
and $s$ a superstring for $\cS_{\cH}$. We refer to $s$
as a normed superstring for $\cS_{\cH}$ if for every $g\in \cH$,
the superstring $s$ contains $s_g$ as a proper substring, whereby $s_g$
is resulted due to a simple alignment of 
the strings included in $S(g)$. 
\end{definition}
\noindent After having  defined  a normed superstring, 
we are going to state  rules which transform a superstring for 
 $\cS_{\cH}$ into a normed superstring for 
 $\cS_{\cH}$ without increasing the length of the underlying superstring.
 All transformation can be performed in polynomial time. 
 Once accomplished to generate a normed superstring, we are able to 
 define the assignment $\psi_s$ and analyze the number of overlapped letters
 in dependence of the number of satisfied equations in $\cH$ by $\psi_s$.
 Let us start with transformations of strings corresponding to circle equations 
 and circle border equations.
\subsubsection{Normalizing Strings Corresponding to Circle and Circle Border Equations}
Let $x_i\oplus x_{i+1}=0$ be a circle equation in $\cH$.  
Furthermore, let   
$x^{m0}_{i} x^{r0}_{i+1} x^{l1}_{i} x^{m1}_{i+1}$ 
and 
$ x^{l1}_{i} x^{m1}_{i+1} x^{m0}_{i} x^{r0}_{i+1}$
be its corresponding strings.
We observe that these strings can have an overlap of at most one letter from the left side 
as well as from the right side with other strings
in $\cS_{\cH}$. Given a superstring $s$ for $\cS_{\cH}$, we obtain at least the same number of 
overlapped letters if we use one of the  simple alignments in $s$. 
In particular, we have to use the simple alignment that maximizes the number of overlapped 
letters.\\
\\
Given a superstring $s$ for $\cS_{\cH}$, we separate the strings 
$x^{m0}_{i} x^{r0}_{i+1} x^{l1}_{i} x^{m1}_{i+1}$
and $ x^{l1}_{i} x^{m1}_{i+1} x^{m0}_{i} x^{r0}_{i+1}$ from $s$. 
Consequently, this results in at most three strings $bx^{m0}_{i}$, $x^{m1}_{i+1} m  x^{l1}_{i}$
and $x^{r0}_{i+1}e$   such that 
$$s=bx^{m0}_{i}x^{r0}_{i+1} x^{l1}_{i}x^{m1}_{i+1} m  x^{l1}_{i}x^{m1}_{i+1} x^{m0}_{i}x^{r0}_{i+1}e. $$ 
Then, we define the transformed superstring $s'$ with at least the same number of overlapped letters by
$$s'=  bx^{m0}_{i}x^{r0}_{i+1} x^{l1}_{i}x^{m1}_{i+1}x^{m0}_{i} x^{r0}_{i+1} e \, x^{m1}_{i+1} m  x^{l1}_{i} . $$
In order to define the simple alignment, which is used in $s$ by the strings in $S(g^x_{i+1})$, we are going to state
a criterion.\\
\\
Let $s$ be a superstring for $\cS_{\cH}$  and  $g^x_{i+1}$ a circle equation. Let the 
 corresponding strings are given by $x^{m0}_{i} x^{r0}_{i+1} x^{l1}_{i} x^{m1}_{i+1}$
and $ x^{l1}_{i} x^{m1}_{i+1} x^{m0}_{i} x^{r0}_{i+1}$. 
Then, we say that the strings in $S(g^x_{i+1})$ use a $1$-alignment in $s$ if there are more strings $s^1$ in 
$\cS_{\cH} \backslash S(g^x_{i+1})$ such that either $s^1$ is overlapped by one letter from the right side with 
$x^{l1}_{i} x^{m1}_{i+1} x^{m0}_{i} x^{r0}_{i+1}$ or $s^1$ is overlapped by one letter from the left side with
$x^{m0}_{i} x^{r0}_{i+1} x^{l1}_{i} x^{m1}_{i+1}$ in $s$ than strings $s^0$ in 
$\cS_{\cH} \backslash S(g^x_{i+1})$ such that either $s^0$ is overlapped by one letter from the left side with 
$x^{l1}_{i} x^{m1}_{i+1} x^{m0}_{i} x^{r0}_{i+1}$ or $s^0$ is overlapped by one letter from the right side with
$x^{m0}_{i} x^{r0}_{i+1} x^{l1}_{i} x^{m1}_{i+1}$ in $s$. Otherwise, the strings in 
$S(g^x_{i+1})$ use a $0$-alignment in $s$.\\
\\
Given a superstring $s$ for $\cS_{\cH}$, we define informally a part of the backbone of our transformed
superstring by  the strings $s_g$, where $s_g$ is resulted due to a simple alignment used in $s$ of    
 the strings $S(g)$ for every circle equation $g\in \cH$. Afterwards, we  use this construction
 to align them with strings corresponding to matching equations, equations with three variables and
 circle border equations. Moreover, it will help us to define the assignment $\psi_s$ and
 relate the number of satisfied equations to the number of overlapped letters. But first, 
 we concentrate on circle border equations.\\
 \\   
 Let $x_1\oplus x_n=0$ be a circle border equation. Furthermore, let the corresponding strings are
given by
$$
L_xC^l_x, \quad  C^l_x x^{m0}_1   x^{l1}_n  C^r_x , \quad 
  x^{l1}_n C^r_x   C^l_x x^{m0}_1 , \quad 
C^l_x x^{r1}_1   x^{m0}_n  C^r_x, \quad 
 x^{m0}_n C^r_x   C^l_x x^{r1}_1, \quad \textrm{ and } \quad
C^r_x R_x.
$$
Since the simple alignments of the strings in $S(g^x_1)$ achieve an overlap of two letters for each pair
$\{C^l_x x^{m0}_1   x^{l1}_n  C^r_x, ~ 
  x^{l1}_n C^r_x   C^l_x x^{m0}_1\}$ and $\{C^l_x x^{r1}_1   x^{m0}_n  C^r_x ,~
 x^{m0}_n C^r_x   C^l_x x^{r1}_1\}$, we
argue as before that these strings can be rearranged in a given  superstring for $\cS_{\cH}$ 
such that the pairs use a simple alignment without increasing the length of the underlying
superstring  for $\cS_{\cH}$. In this situation, we are able to overlap one of the  pairs using a simple alignment
 with $L_xC^l_x$  from the left side
and the other one with $C^r_x R_x$ from the right  side without increasing the length. This construction
 checks whether the variables $x_1$ and $x_n$ have the same assigned value, which 
is rewarded by another overlap of one letter of the corresponding  strings using a simple alignment. \\
\\    
For any fixed order of the circles $C^x$ in $\cH$, we build the backbone of our superstring
consisting of the concatenation of the strings $s^xs^y\cdots s^z$, where the string $s^x$
is associated to its circle $C^x$. Furthermore, $s^x$ consists of the corresponding 
simple alignments of the strings in $S(g^x_i)$ used in $s$ and the order of the strings
is given by the order of the variables in $C^x$. The string $s^x$ starts with the letter $L_x$
and ends with $R_x$.\\
\\  
Notice that similar transformations can be applied to
 strings corresponding to matching equations and to equations with three variables, 
 but we are going to define the transformation for those strings in detail while 
 analyzing the upper bound of overlapped letters for simple aligned
 strings corresponding to circle equations, which  are contained in a given
  superstring $s$ for $\cS_{\cH}$.  \\
  \\
  Before we start our  analysis, we define the assignment $\psi_{s}$ 
  based on the actual superstring $s$ for $\cS_{\cH}$, which is not necessarily
  a normed superstring for $\cS_{\cH}$. By applying the transformations,
  which we are going to define, the assignment  $\psi_s$ will change 
  in dependence to the actual considered superstring. 
  \begin{eqnarray*}
  \psi_s(x_i) &=& 1 \textrm{ if the strings in $S(g^x_i)$ use a $1$-alignment in $s$}\\
              & =& 0  \textrm{ otherwise } 
  \end{eqnarray*}
  Due to the transformations for the strings corresponding to circle and circle border
  equations, the assignment $\psi_s$ is well-defined.
\subsubsection{Defining the Assignment for  Checker Variables}
Let $x\in V(\mathcal{E}_3)$, $C^x$ be the corresponding circle in $\cH$
and $M^x$ its associated perfect matching. Furthermore, let
$x_{i}\oplus x_{i+1}= 0$, $x_{i-1}\oplus x_{i}= 0$, $x_{j-1}\oplus x_{j}= 0$,
$x_{j}\oplus x_{j+1}= 0$
and $x_i \oplus x_j =0$ be equations in $\cH$, where $\{i,j\}\in M^x$
and $i<j$
 holds. Let $s$ be a superstring for $\cS_{\cH}$ such that the strings
 corresponding to circle and circle border equations are using a
 simple alignment in $s$. 
Based on the  simple alignments of the strings corresponding to
$g^x_{i}$, $g^x_{i+1}$, $g^x_{j}$ and $g^x_{j+1}$, which are used in the 
superstring $s$, we are going to define the assignment to
the variables $x_i$ and $x_j$.
Furthermore, we analyze the number of overlapped letters that can be achieved 
 given the simple aligned strings and relate them  to the number of satisfied
equations in $\cH$ by $\psi_s$. \\
\\
In the remainder, we will assume that the underlying
 superstring for $\cS_{\cH}$ contains simple aligned strings
 corresponding to circle and circle border equations. 
Before we start our analysis, we introduce the notation of a constellation
that  denotes which of the simple alignments are used by the strings corresponding 
to the equations $g^x_i$,  
$g^x_{i+1}$, $g^x_j$ and $g^x_{j+1}$ in $s$.\\
\\
Given a superstring $s$ for $S_{\cH}$ and $\{i,j\}\in M^x$, a
 \emph{constellation} $c$ is defined by $(X_{i}X_{i+1},X_{j-1}X_{j+1})^s_{\{i,j\}}$ with
 $X_{i},X_{i+1},X_{j},X_{j+1}\in \{0,1\}$, where $X_k=1$ if and only if   
the  strings in $S(g^x_{k})$ use
the $1$-alignment in $s$ for $k\in \{i,i+1,j,j+1\}$. 
We call a constellation $c$ inconsistent if there is an  entry $A_1A_2$ with
 $A_1\neq A_2$. Otherwise, $c$ is called  consistent.\\
 \\ 
Based on the given constellations, we are going to define $\psi_s$. 
\begin{definition}[Assignment $\psi_s$ to Checker Variables]
Let $\cH$ be an instance of the  Hybrid problem, $\cS_{\cH}$ its corresponding
instance of the superstring problem and $s$  a superstring for $\cS_{\cH}$.
Given the constellation $(X_{i}X_{i+1},X_{j}X_{j+1})^s_{\{i,j\}}$, we define $\psi_s$ 
in the following way.
\begin{enumerate}
\item[$(i)$] $\psi_s(x_i)=X_i$ and $\psi_s(x_j)=X_j$ if
 $X_i  \oplus X_j=1 $ and  $c$ is  consistent
\item[$(ii)$] $\psi_s(x_i)=X_i$ and $\psi_s(x_j)=X_j$ if
 $X_i  \oplus X_j=0 $ 
\item[$(iii)$] $\psi_s(x_i)=1-X_i$ and $\psi_s(x_j)=X_j$ if
 $X_i  \oplus X_j=1 $ and $X_i\neq X_{i+1}$
\item[$(iv)$] $\psi_s(x_i)=X_i$ and $\psi_s(x_j)=1-X_j$ if
 $X_i  \oplus X_j=1 $, $X_j\neq X_{j+1}$ and $X_i= X_{i+1}$
\end{enumerate}
\end{definition}
\noindent
We are going to analyze the the different constellations and discuss the cases 
$(i)$-$(iv)$ of the definition of $\psi_s$. We start with case $(i)$.\\
\\
%
\textbf{CASE $(i)$~$X_i  \oplus X_j=1 $ and  $c$ is consistent:}\\
%
%
There are two constellations, which we have to analyze, namely $(11,00)^s_{\{i,j\}}$ and 
$(00,11)^s_{\{i,j\}}$. Starting with the former constellation, we obtain the  scenario depicted 
below.
The string $\medtriangleright\medtriangleright \framebox{$X_i$}$ with $X_i\in \{x^{m0}_i,x^{r1}_i\}$ 
represents a simple alignment of the strings in $S(g^x_{i})$.
Analogously, the string 
$$\framebox{$X_{i+1}$} \medtriangleleft\medtriangleleft \textrm{ with } X_{i+1}\in \{x^{m0}_i,x^{l1}_i\}$$ 
represents a simple alignment of the strings in $S(g^x_{i+1})$.
Since we know that using the most profitable simple alignment of the strings in $S(g^x_{\{i,j\}})$ 
 does not increase the length of the superstring, we make use of the $1$-alignment and transform
 the superstring $s$ in the superstring $s'$, which are both depicted below.   
\\
\vspace{3mm}\\
\fbox{\parbox{\dimexpr \linewidth - 2\fboxrule - 2\fboxsep}{
\vspace{4mm}
$$
s=b
 \medtriangleright\medtriangleright
 \framebox{$x^{r1}_i$}\,
 x^{r1}_i x^{l1}_i x^{r0}_j x^{l0}_j \,
\framebox{$x^{l1}_i$}   \medtriangleleft \medtriangleleft 
m
\medtriangleright\medtriangleright
\framebox{$x^{r0}_j$} \, x^{r0}_j x^{l0}_j x^{r1}_i x^{l1}_i \,
\framebox{$x^{l0}_j$} 
 \medtriangleleft \medtriangleleft
 e
$$
$$\downarrow$$
$$
s'=b
\medtriangleright\medtriangleright
\lefteqn{ \overbrace{ \phantom{\framebox{$x^{r1}_i$}x^{l1}_i x^{r0}_j
x^{l0}_j }}^{  x^{r1}_i x^{l1}_i x^{r0}_j x^{l0}_j   } }
 \framebox{$x^{r1}_i$}\, x^{l1}_i 
\lefteqn{ \underbrace{ \phantom{x^{r0}_j x^{l0}_j x^{r1}_i  \framebox{$x^{l1}_i$} }}
_{x^{r0}_j x^{l0}_j x^{r1}_i x^{l1}_i }  }
x^{r0}_j x^{l0}_j x^{r1}_i \framebox{$x^{l1}_i$}
  \medtriangleleft \medtriangleleft 
m \,
   \medtriangleright\medtriangleright \framebox{$x^{r0}_j$} \,
\framebox{$x^{l0}_j$}
\medtriangleleft \medtriangleleft
e
$$\vspace{1mm}
}}
\vspace{1mm}\\
Let us derive an upper bound on the number of overlapped  letters. More precisely, we 
are interested in 
the number of overlapped  letters being additional 
to the overlap of two letters due to the simple alignment. 
In both cases, either by using the $1$-alignment or the $0$-alignment of the strings in $S(g^x_{\{i,j\}})$,
we cannot obtain more than an overlap of two letters.
It corresponds to the number of  satisfied equations, which are 
$x_i\oplus x_{i+1}=0$ and  $x_j\oplus x_{j+1}=0$.\\
\\
In case of the constellation $(00,11)^s_{\{i,j\}}$, we separate the 
strings $x^{r1}_i x^{l1}_i x^{r0}_j x^{l0}_j$ and $x^{r0}_j x^{l0}_j x^{r1}_i x^{l1}_i$ from the superstring
$s$. Then, we attach the aligned string $x^{r1}_i x^{l1}_i x^{r0}_j x^{l0}_jx^{r1}_i x^{l1}_i$ 
at the end of the actual solution. The considered
situation is depicted below.\\
\vspace{3mm}\\
\fbox{\parbox{\dimexpr \linewidth - 2\fboxrule - 2\fboxsep}{
\vspace{4mm}
$$
b
\medtriangleright\medtriangleright
\framebox{$x^{m0}_i$} \,
x^{r1}_i x^{l1}_i x^{r0}_j x^{l0}_j \,
\framebox{$x^{m0}_i$}
\medtriangleleft \medtriangleleft
m
\medtriangleright\medtriangleright
\framebox{$x^{m1}_j$}  \,
 x^{r0}_j x^{l0}_j x^{r1}_i x^{l1}_i\,
\framebox{$x^{m1}_j$}
\medtriangleleft \medtriangleleft
e
$$
$$\downarrow$$
$$
b
\lefteqn{ \overbrace{\phantom{ 
\medtriangleright\medtriangleright \framebox{$x^{m0}_i$}  }  }
^{ \medtriangleright\medtriangleright \framebox{$x^{m0}_i$} }  }
\medtriangleright\medtriangleright
\framebox{$x^{m0}_i$} 
\medtriangleleft \medtriangleleft m\,
\lefteqn{ \overbrace{\phantom{ \medtriangleright\medtriangleright   \framebox{$x^{m1}_j$}    }  } 
^{\medtriangleright\medtriangleright   \framebox{$x^{m1}_j$}}   }
\medtriangleright\medtriangleright
\framebox{$x^{m1}_j$} 
\medtriangleleft \medtriangleleft
e
 \quad
\lefteqn{ \overbrace{\phantom{ x^{r1}_i x^{l1}_i x^{r0}_j x^{l0}_j  }  }^
{x^{r1}_i x^{l1}_i x^{r0}_j x^{l0}_j}  } 
x^{r1}_i x^{l1}_i 
\lefteqn{ \underbrace{\phantom{ x^{r0}_j x^{l0}_j x^{r1}_i x^{l1}_i  }  }_
{x^{r0}_j x^{l0}_j x^{r1}_i x^{l1}_i}  } 
x^{r0}_j x^{l0}_j x^{r1}_i x^{l1}_i
$$
\vspace{1mm}
}}
\vspace{1mm}\\
In this scenario, the best that we are able to obtain is an overlap of two letters. This corresponds
to the number of satisfied equations, namely $x_i\oplus x_{i+1}=0$ and  $x_j\oplus x_{j+1}=0$.\\
\\
\textbf{CASE $(ii)$ ~ $ X_i  \oplus X_j  =0 $ :}\\
Let us start with the constellation $(0X_{i+1},0X_{j+1})^s_{\{i,j\}}$.
In this case, we set  $\psi_s(x_i)=0$ and $\psi_s(x_j)=0$.
Given the strings $\medtriangleright\medtriangleright\framebox{$x^{m0}_i$}$\, , 
$\framebox{$X_{i+1}$}\medtriangleleft \medtriangleleft$,
$\medtriangleright\medtriangleright \framebox{$x^{r0}_j$}$ and 
$ \framebox{$X_{j+1}$}\medtriangleleft\medtriangleleft$ with $X_{i+1}\in \{x^{m0}_i, x^{l1}_i\}$ 
and $X_{j+1} \in \{x^{m1}_j, x^{l0}_j\}$,
we obtain the following scenario:\\
\vspace{3mm}\\
\fbox{\parbox{\dimexpr \linewidth - 2\fboxrule - 2\fboxsep}{
\vspace{4mm}
$$
b
\medtriangleright\medtriangleright 
 \framebox{$x^{m0}_i$} \,
 x^{r1}_i x^{l1}_i x^{r0}_j x^{l0}_j \,
 \framebox{$X_{i+1}$}\medtriangleleft\medtriangleleft
 m
 \medtriangleright\medtriangleright
\lefteqn{ \overbrace{\phantom{ \framebox{$x^{r0}_j$} x^{l0}_j x^{r1}_i x^{l1}_i }}
^{x^{r0}_j x^{l0}_j x^{r1}_i x^{l1}_i} }
\framebox{$x^{r0}_j$}  x^{l0}_j x^{r1}_i x^{l1}_i
\framebox{$X_{j+1}$}\medtriangleleft \medtriangleleft  
e 
$$
$$\downarrow$$
$$
b
\medtriangleright\medtriangleright 
 \framebox{$x^{m0}_i$} \,
\framebox{$X_{i+1}$}\medtriangleleft\medtriangleleft
 m
 \medtriangleright\medtriangleright
\lefteqn{ \overbrace{\phantom{ \framebox{$x^{r0}_j$} x^{l0}_j x^{r1}_i x^{l1}_i }}
^{x^{r0}_j x^{l0}_j x^{r1}_i x^{l1}_i} }
\framebox{$x^{r0}_j$}  x^{l0}_j 
\lefteqn{ \underbrace{\phantom{ x^{r1}_i x^{l1}_i x^{r0}_j x^{l0}_j  }}
_{x^{r1}_i x^{l1}_i x^{r0}_j x^{l0}_j   } }
x^{r1}_i x^{l1}_ix^{r0}_j x^{l0}_j
\framebox{$X_{j+1}$}\medtriangleleft \medtriangleleft  
e 
$$\vspace{1mm}
}}
\vspace{1mm}\\
The most advantageous simple alignment in this case is the  $0$-alignment of the strings in $S(g^x_{\{i,j\}})$. 
If  $\psi_s(x_i)=\psi_s(x_{i+1})=0$
holds, which means $X_{i+1}=x^{m0}_i$, we obtain another overlap of one letter by aligning 
$\medtriangleright\medtriangleright 
 \framebox{$x^{m0}_i$}$ with 
$\framebox{$x^{m0}_i$}\medtriangleleft \medtriangleleft$.
A similar argument holds for $\psi_s(x_{j})=\psi_s(x_{j+1})=0$. Notice that the equation $x_{i}\oplus x_{j}=0$
is satisfied by $\psi_s$.
In summary, we state that we obtain an  overlap of one additional letter per satisfied equation.
Hence, we obtain an overlap of three letters according to the satisfied equations
$x_{i}\oplus x_{i+1}=0$,  $x_{i}\oplus x_{j}=0$ and $x_{j}\oplus x_{j+1}=0$.\\
\\
Consider  the constellation $(1X_{i+1},1X_{j+1})^s_{\{i,j\}}$. Hence, we are
given the strings $\medtriangleright\medtriangleright\framebox{$x^{r1}_i$}$\, , 
$\framebox{$X_{i+1}$}\medtriangleleft \medtriangleleft$,
$\medtriangleright\medtriangleright \framebox{$x^{m1}_j$}$ and 
$ \framebox{$X_{j+1}$}\medtriangleleft\medtriangleleft$ with $X_{i+1}\in \{x^{m0}_i, x^{l1}_i\}$
and $X_{j+1} \in \{x^{m1}_j, x^{l0}_j\}$.
We obtain the scenario displayed below.\\
\vspace{3mm}\\
\fbox{\parbox{\dimexpr \linewidth - 2\fboxrule - 2\fboxsep}{
\vspace{4mm}
$$
b
\medtriangleright\medtriangleright \framebox{$x^{r1}_i$} \,
x^{r1}_i x^{l1}_i x^{r0}_j x^{l0}_j \,
  \framebox{$X_{i+1}$}\medtriangleleft \medtriangleleft
  m
\medtriangleright\medtriangleright\framebox{$x^{m1}_j$} \,
x^{r0}_j x^{l0}_j   x^{r1}_i x^{l1}_i \,
 \framebox{$X_{j+1}$} \medtriangleleft \medtriangleleft
 e
$$
$$\downarrow$$
$$
b
\medtriangleright\medtriangleright
\lefteqn{ \overbrace{ \phantom{\framebox{$x^{r1}_i$}x^{l1}_i x^{r0}_j
x^{l0}_j }}^{  x^{r1}_i x^{l1}_i x^{r0}_j x^{l0}_j   } }
 \framebox{$x^{r1}_i$}\, x^{l1}_i 
\lefteqn{ \underbrace{ \phantom{x^{r0}_j x^{l0}_j x^{r1}_i  x^{l1}_i }}
_{x^{r0}_j x^{l0}_j x^{r1}_i x^{l1}_i }  }
x^{r0}_j x^{l0}_j x^{r1}_i   x^{l1}_i
 \framebox{$X_{i+1}$}
  \medtriangleleft \medtriangleleft m
 \medtriangleright\medtriangleright\framebox{$x^{m1}_j$} \,
 \framebox{$X_{j+1}$} \medtriangleleft \medtriangleleft
 e
$$
\vspace{1mm}
}}
\vspace{1mm}\\
In this case, we use the $1$-alignment of the strings in $S(g^x_{\{i,j\}})$. If  $\psi_s(x_i)=\psi_s(x_{i+1})=1$
holds, which means $X_{i+1}=x^{l1}_i$, we obtain another overlap of one letter by aligning 
$$
\medtriangleright\medtriangleright
\lefteqn{ \overbrace{ \phantom{\framebox{$x^{r1}_i$}x^{l1}_i x^{r0}_j
x^{l0}_j }}^{  x^{r1}_i x^{l1}_i x^{r0}_j x^{l0}_j   } }
 \framebox{$x^{r1}_i$}\, x^{l1}_i 
\lefteqn{ \underbrace{ \phantom{x^{r0}_j x^{l0}_j x^{r1}_i  x^{l1}_i }}
_{x^{r0}_j x^{l0}_j x^{r1}_i x^{l1}_i }  }
x^{r0}_j x^{l0}_j x^{r1}_i   x^{l1}_i 
\quad 
\textrm{ with } 
\quad
\framebox{$x^{l1}_i$}\medtriangleleft \medtriangleleft.
$$
In case of $\psi_s(x_{j})=\psi_s(x_{j+1})=1$, we may apply a similar argument. Notice that the equation $x_{i}\oplus x_{j}=0$
is satisfied by $\psi_s$.
In summary, we state that we obtain an  overlap of one additional letter per satisfied equation.
Hence, we obtain an overlap of three letters according to the satisfied equations
$x_{i}\oplus x_{i+1}=0$,  $x_{i}\oplus x_{j}=0$ and $x_{j}\oplus x_{j+1}=0$.\\
\\
\textbf{CASE $(iii)$~ $X_i  \oplus X_j=1 $ and $X_i\neq X_{i+1}$:}\\
Let us begin with the constellation $(10,0X_{j+1})^s_{\{i,j\}}$. 
We consider the  scenario depicted below, in which we are given
the strings $\medtriangleright\medtriangleright\framebox{$x^{r1}_i$}$\, , 
$\framebox{$x^{m0}_i$}\medtriangleleft \medtriangleleft$,
$\medtriangleright\medtriangleright \framebox{$x^{l0}_j$}$ and 
$ \framebox{$X_{j+1}$}\medtriangleleft\medtriangleleft$ with $X_{j+1}\in \{x^{l0}_j, x^{m1}_j\}$.\\
\vspace{3mm}\\
\fbox{\parbox{\dimexpr \linewidth - 2\fboxrule - 2\fboxsep}{
\vspace{4mm}
$$
b
\medtriangleright\medtriangleright
\framebox{$x^{r1}_i$} \,
x^{r1}_i x^{l1}_i x^{r0}_j x^{l0}_j \,
 \framebox{$x^{m0}_i$} \medtriangleleft\medtriangleleft 
 m 
\medtriangleright\medtriangleright\framebox{$x^{r0}_j$} \,
x^{r0}_j x^{l0}_j x^{r1}_i x^{l1}_i \,
\framebox{$X_{j+1}$}
\medtriangleleft\medtriangleleft
e
$$
$$\downarrow$$
$$
b
\lefteqn{ \overbrace{\phantom{ 
\medtriangleright\medtriangleright \framebox{$x^{m0}_i$}  }  }
^{ \medtriangleright\medtriangleright \framebox{$x^{m0}_i$} }  }
\medtriangleright\medtriangleright
\framebox{$x^{m0}_i$} 
\medtriangleleft \medtriangleleft m\,
\medtriangleright\medtriangleright
\lefteqn{ \overbrace{\phantom{ \framebox{$x^{r0}_j$} x^{l0}_j x^{r1}_i x^{l1}_i }}
^{x^{r0}_j x^{l0}_j x^{r1}_i x^{l1}_i} }
\framebox{$x^{r0}_j$}  x^{l0}_j 
\lefteqn{ \underbrace{\phantom{ x^{r1}_i x^{l1}_i x^{r0}_j x^{l0}_j  }}
_{x^{r1}_i x^{l1}_i x^{r0}_j x^{l0}_j   } }
x^{r1}_i x^{l1}_ix^{r0}_j x^{l0}_j \,
\framebox{$X_{j+1}$}\medtriangleleft \medtriangleleft  
e 
$$
\vspace{1mm}
}}
\vspace{1mm}\\
Instead of using the $1$-alignment of the strings in  $S(g^x_i)$, we rather switch to
the $0$-alignment, i.e. we obtain the string 
$\medtriangleright\medtriangleright \framebox{$x^{m0}_i$}$
and define $\psi(x_i)=0$.
It results directly in gaining two additional satisfied equations and 
an overlap of one additional letter. As a matter of fact,
we might lose an overlap of one letter, because the string
 $\medtriangleright\medtriangleright\framebox{$x^m_1$}$ might have been aligned
 from the right side with another string. Furthermore, the equation $x_{i-1}\oplus x_{i}=0$
 might be unsatisfied. But all in all, we obtain at least $2-1$ additional satisfied equations
 by switching the value without increasing the superstring.  
Notice that we may achieve an additional overlap of one letter if $X_{j+1}=x^{l0}_j$ holds,
which means that $\psi_s$ satisfies the equation $x_j\oplus x_{j+1}=0$.\\
\\
The next  constellation, we are going to analyze, is $(01,1X_{j+1})^s_{\{i,j\}}$. 
Hence, we are given
the strings $\medtriangleright\medtriangleright\framebox{$x^{m0}_i$}$\, , 
$\framebox{$x^{l1}_i$}\medtriangleleft \medtriangleleft$,
$\medtriangleright\medtriangleright \framebox{$x^{m1}_j$}$ and 
$ \framebox{$X_{j+1}$}\medtriangleleft\medtriangleleft$, with $X_{j+1}\in \{x^{l0}_j, x^{m1}_j\}$.
The situation is displayed below.\\
\vspace{3mm}\\
\fbox{\parbox{\dimexpr \linewidth - 2\fboxrule - 2\fboxsep}{ 
\vspace{4mm}
$$
b
\medtriangleright\medtriangleright
\framebox{$x^{m0}_i$} \,
x^{r1}_i x^{l1}_i x^{r0}_j x^{l0}_j \,
 \framebox{$x^{l1}_i$} \medtriangleleft\medtriangleleft 
 m 
 \medtriangleright\medtriangleright\framebox{$x^{m1}_j$}  \,
x^{r0}_j x^{l0}_j x^{r1}_i x^{l1}_i \,
\framebox{$X_{j+1}$}
\medtriangleleft\medtriangleleft
e
 $$
$$\downarrow$$
$$
b
\medtriangleright\medtriangleright
\lefteqn{ \overbrace{ \phantom{\framebox{$x^{r1}_i$}x^{l1}_i x^{r0}_j
x^{l0}_j }}^{  x^{r1}_i x^{l1}_i x^{r0}_j x^{l0}_j   } }
 \framebox{$x^{r1}_i$}\, x^{l1}_i 
\lefteqn{ \underbrace{ \phantom{x^{r0}_j x^{l0}_j x^{r1}_i  \framebox{$x^{l1}_i$} }}
_{x^{r0}_j x^{l0}_j x^{r1}_i x^{l1}_i }  }
x^{r0}_j x^{l0}_j x^{r1}_i   
\framebox{$x^{l1}_i$}
  \medtriangleleft \medtriangleleft m
 \medtriangleright\medtriangleright\framebox{$x^{m1}_j$}
 ~\framebox{$X_{j+1}$} \medtriangleleft \medtriangleleft
 e
 $$
 \vspace{1mm}
 }}
 \vspace{1mm}\\
We obtain a similar situation, in which we switch 
$\medtriangleright\medtriangleright\framebox{$x^{m0}_i$}$
to 
$\medtriangleright\medtriangleright \framebox{$x^{r1}_i$}$.
Accordingly, we define $\psi_s(x_i)=1$.
We obtain at least one additional satisfied equation
 by switching the value without increasing the length of the superstring.  
Notice that we may achieve an additional overlap of one letter if $X_{j+1}=x^{m1}_j$ holds.
It corresponds to the satisfied   equation $x_j\oplus x_{j+1}=0$.\\
\\
\textbf{CASE $(iv)$~ $X_i  \oplus X_j=1 $, $X_j\neq X_{j+1}$ and $X_i= X_{i+1}$:}\\
Starting our analysis with the constellation $(00,10)^s_{\{i,j\}}$, we obtain the following scenario.\\
\vspace{1mm}\\
\fbox{\parbox{\dimexpr \linewidth - 2\fboxrule - 2\fboxsep}{ 
\vspace{4mm}
$$
b
\medtriangleright\medtriangleright
\framebox{$x^{m0}_i$} \,
x^{r1}_i x^{l1}_i x^{r0}_j x^{l0}_j \,
 \framebox{$x^{m0}_i$} \medtriangleleft\medtriangleleft 
 m 
 \medtriangleright\medtriangleright\framebox{$x^{m1}_j$} \,
x^{r0}_j x^{l0}_j x^{r1}_i x^{l1}_i \,
\framebox{$x^{l0}_j$}
\medtriangleleft\medtriangleleft
e
 $$
$$\downarrow$$
$$
b
\lefteqn{ \overbrace{\phantom{ \medtriangleright\medtriangleright 
 \framebox{$x^{m0}_i$} }}
^{\medtriangleright\medtriangleright 
 \framebox{$x^{m0}_i$} } }
\medtriangleright\medtriangleright 
 \framebox{$x^{m0}_i$} \,
\medtriangleleft\medtriangleleft
 m
 \medtriangleright\medtriangleright
\lefteqn{ \overbrace{\phantom{ \framebox{$x^{r0}_j$} x^{l0}_j x^{r1}_i x^{l1}_i }}
^{x^{r0}_j x^{l0}_j x^{r1}_i x^{l1}_i} }
\framebox{$x^{r0}_j$}  x^{l0}_j 
\lefteqn{ \underbrace{\phantom{ x^{r1}_i x^{l1}_i x^{r0}_j \framebox{$ x^{l0}_j$}  }}
_{x^{r1}_i x^{l1}_i x^{r0}_j x^{l0}_j   } }
x^{r1}_i x^{l1}_ix^{r0}_j 
\framebox{$ x^{l0}_j$}\medtriangleleft \medtriangleleft  
e 
 $$\vspace{1mm}
 }}
 \vspace{3mm}\\
In this case, we argue that we switch the string $\medtriangleright\medtriangleright\framebox{$x^{m1}_j$}$
to  $\medtriangleright\medtriangleright\framebox{$x^{r0}_j$}$. This means that we set $\psi_s(x_j)=0$.
This transformation yields an overlap of at least the same number of letters, since we might lose an overlap
of one letter from the left side. On the other hand, we align the string  
$$
\medtriangleright\medtriangleright\framebox{$x^{r0}_j$} \quad \textrm{ with } \quad 
\lefteqn{ \overbrace{\phantom{ x^{r0}_j x^{l0}_j x^{r1}_i x^{l1}_i }}
^{x^{r0}_j x^{l0}_j x^{r1}_i x^{l1}_i} }
x^{r0}_j  x^{l0}_j 
\lefteqn{ \underbrace{\phantom{ x^{r1}_i x^{l1}_i x^{r0}_j \framebox{$ x^{l0}_j$}  }}
_{x^{r1}_i x^{l1}_i x^{r0}_j x^{l0}_j   } }
x^{r1}_i x^{l1}_ix^{r0}_j 
\framebox{$ x^{l0}_j$}\medtriangleleft \medtriangleleft
$$
 from the right side by one letter. Notice that we  gain at least one additional satisfied equation. \\
\\
The last constellation, we are going to analyze, is 
$(11,01)^s_{\{i,j\}}$.  The corresponding situation is depicted below. \\
\vspace{1mm}\\
\fbox{\parbox{\dimexpr \linewidth - 2\fboxrule - 2\fboxsep}{ 
\vspace{4mm}
$$
b
\medtriangleright\medtriangleright
\framebox{$x^{r1}_i$} \,
x^{r1}_i x^{l1}_i x^{r0}_j x^{l0}_j \,
 \framebox{$x^{l1}_i$} \medtriangleleft\medtriangleleft 
 m 
 \medtriangleright\medtriangleright\framebox{$x^{r0}_j$} \,
x^{r0}_j x^{l0}_j x^{r1}_i x^{l1}_i \,
\framebox{$x^{m1}_j$}
\medtriangleleft\medtriangleleft
e
 $$
$$\downarrow$$
$$
b
\medtriangleright\medtriangleright
\lefteqn{ \overbrace{ \phantom{\framebox{$x^{r1}_i$}x^{l1}_i x^{r0}_j
x^{l0}_j }}^{  x^{r1}_i x^{l1}_i x^{r0}_j x^{l0}_j   } }
 \framebox{$x^{r1}_i$}\, x^{l1}_i 
\lefteqn{ \underbrace{ \phantom{x^{r0}_j x^{l0}_j x^{r1}_i  \framebox{$x^{l1}_i$} }}
_{x^{r0}_j x^{l0}_j x^{r1}_i x^{l1}_i }  }
x^{r0}_j x^{l0}_j x^{r1}_i   
\framebox{$x^{l1}_i$}
  \medtriangleleft \medtriangleleft m
  \lefteqn{ \overbrace{ \phantom{\medtriangleright\medtriangleright\framebox{$x^{m1}_j$} }}
  ^{  \medtriangleright\medtriangleright\framebox{$x^{m1}_j$}   } }
 \medtriangleright\medtriangleright\framebox{$x^{m1}_j$}
  \medtriangleleft \medtriangleleft
 e
 $$
 \vspace{1mm}
 }}
 \vspace{3mm}\\
 In this case, we switch the string $\medtriangleright\medtriangleright\framebox{$x^{r0}_j$}$
 to $\medtriangleright\medtriangleright\framebox{$x^{m1}_j$}$. Similarly to the former case,
 this transformation does not increase the length of the superstring. By defining $\psi_s(x_j)=1$, we achieve at least 
 one more satisfied equation.\\
 \\
 In summary, we note that we achieve at least the same number of satisfied equations as the number of overlapped 
 letters. By applying the defined transformations, the actual  
superstring contains only strings corresponding to
 matching equations using a simple alignment.  
 Matching equations $x_i\oplus x_j=0$ with $i>j$ can be analyzed analogously.\\
\\
 We are going to define the assignment for contact variables. 
 
\subsubsection{Defining the Assignment for Contact Variables }
Let $g^3_j\equiv x\oplus y\oplus z =0$ be an equation with exactly three
variables in $\cH$. 
Given a simple alignment of the strings corresponding to the equations $x_{j_1-1}\oplus x=0$,  
$x \oplus   x_{j_1+1}=0$, $y_{j_2-1}\oplus y=0$,  
$y \oplus   y_{j_2+1}=0$, $z_{j_3-1}\oplus z=0$, and 
$z \oplus   z_{j_3+1}=0$, we are going to define an assignment based
on the underlying simple alignments and analyze the number of satisfied equations
in dependence of the number of overlapped letters in the superstring.\\
\\
For a given superstring $s$ for $\cS_{\cH}$ and equation 
$g^3_j\equiv x\oplus y \oplus z =0$, we define a \emph{constellation} 
$c$ given by  $(X_1X_2,Y_1Y_2,Z_1Z_2)^s_j$
with $X_1,X_2,Y_1,Y_2,Z_1,Z_2\in \{0,1\}$, where $C=1$ with 
$C\in \{X_1,X_2,Y_1,Y_2,Z_1,Z_2\}$ if and only if the 
strings in the corresponding set are using a $1$-alignment in $s$.
 A constellation denotes which of the simple alignments 
is used by the strings in $s$.  
We call a constellation inconsistent if there is an 
entry $A_1A_2$ such that $A_1\neq A_2$. Otherwise, $c$ is called
consistent.\\
\\
Based on a constellation for a given superstring and an equation $g^3_j$ with three variables,
we are going to define the assignment $\psi_s$ for the variables in $g^3_j$.
\begin{definition}[Assignment $\psi_s$ to Contact Variables]
Let $\cH$ be an instance of the  Hybrid problem, $\cS_{\cH}$ its corresponding
instance of the superstring problem, $s$  a superstring for $\cS_{\cH}$ and
$g^3_j\equiv x\oplus y \oplus z =0$ an equation with three variables
in $\cH$.
For the associated constellation $c=(X_1X_2,Y_1Y_2,Z_1Z_2)^s_j$, 
we define $\psi_s$ 
in the following way.

\begin{enumerate}
\item[$(i)$] If $c$ is  consistent, then, we define 
$\psi_s(x)=X_1$, $\psi_s(y)=Y_1$ 
and $\psi_s(z)=Z_1$ 
\item[$(ii)$] Otherwise, let $A_1A_2$ be an entry in $c$ 
with $A_1\neq A_2$ and $\alpha $ its corresponding variable.
Furthermore, let  $\beta$ and $\gamma$ be  variables associated with the entry $B_1B_2$
and $C_1C_2$, respectively. If  $A_1  \oplus B_1  \oplus  C_1=0$ holds, 
we define  $\psi_s(\alpha )=A_1$, $\psi_s(\beta)=B_1$ 
and $\psi_s(\gamma)=C_1$. 
\item[$(iii)$] Otherwise, we have $A_1  \oplus B_1  \oplus  C_1=1$. Then, we define
$\psi_s(\alpha )=1-A_1$, $\psi_s(\beta)=B_1$ 
and $\psi_s(\gamma)=C_1$. 
\end{enumerate}
\end{definition}
We are going to analyze the following three cases and define the transformations
for the actual superstring for $\cS_{\cH}$.
\begin{enumerate}
\item[$(i)$] $X_1  \oplus Y_1  \oplus  Z_1=1 $ and $c$ is consistent
\item[$(ii)$] $X_1  \oplus Y_1  \oplus  Z_1=0 $ and $c$ is inconsistent
\item[$(iii)$] $X_1  \oplus Y_1  \oplus  Z_1=1 $ and $c$ is inconsistent
\end{enumerate}
Let us begin with case $(i)$.\\
\\
\textbf{CASE $(i)$ ~$X_1  \oplus Y_1  \oplus  Z_1=1 $ and $c$ is consistent:}\\
In this case, we start with the constellation $(11,11,11)^s_j$. We depict the considered
situation below.\\
\vspace{3mm}\\
\fbox{\parbox{\dimexpr \linewidth - 2\fboxrule - 2\fboxsep}{\
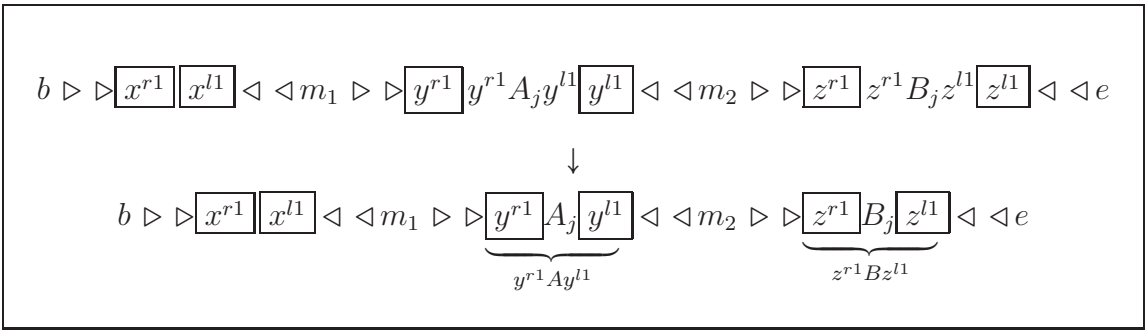
\vspace{2mm}
$$
b
\medtriangleright
\medtriangleright  
\framebox{$x^{r1}$}\,
\framebox{$x^{l1}$} 
\medtriangleleft\medtriangleleft
m_1
\medtriangleright   
\medtriangleright
  \framebox{$y^{r1}$}\,
y^{r1}A_j y^{l1} 
 \framebox{$y^{l1}$}
\medtriangleleft\medtriangleleft
m_2
\medtriangleright   
\medtriangleright
  \framebox{$z^{r1}$}\, 
z^{r1}B_j z^{l1}
 \framebox{$z^{l1}$}
\medtriangleleft\medtriangleleft
e
$$
$$
\downarrow
$$
$$
b
\medtriangleright
\medtriangleright  
\framebox{$x^{r1}$} \,
\framebox{$x^{l1}$} 
\medtriangleleft\medtriangleleft
m_1
\medtriangleright   
\medtriangleright
\lefteqn{  \underbrace{ \phantom{   \framebox{$y^{r1}$} A \framebox{$y^{l1}$}  }}
_{y^{r1}A y^{l1} }  }
\framebox{$y^{r1}$} A_j \framebox{$y^{l1}$}
\medtriangleleft\medtriangleleft
m_2
\medtriangleright   
\medtriangleright
\lefteqn{  \underbrace{ \phantom{   \framebox{$z^{r1}$} B \framebox{$z^{l1}$}  }}
_{z^{r1}B z^{l1} }  }
\framebox{$z^{r1}$} B_j \framebox{$z^{l1}$}
\medtriangleleft\medtriangleleft
e
$$
\vspace{1mm}
}}
\vspace{1mm}\\
According to the definition of $\psi_s$, we have $\psi_s(x)=\psi_s(y)=\psi_s(z)=1$. Notice that 
the equation $x\oplus y \oplus z =0$ is unsatisfied. On the other hand, 
the assignment $\psi_s$
satisfies the equations $x\oplus x_{j_1+1}=0$, $y\oplus y_{j_2+1}=0$ and 
$z\oplus z_{j_3+1}=0$.\\
We note that a string corresponding to $S^A(g^3_j)$ or $S^B(g^3_j)$ using a simple alignment  
can have an overlap of at most
one letter from the right side as well as from the left side. Therefore,  the best 
we can hope for is to overlap the string
$y^{r1}A y^{l1}$ with 
$\medtriangleright   
\medtriangleright
  \framebox{$y^{r1}$}$ and 
$ \framebox{$y^{l1}$}
\medtriangleleft\medtriangleleft$
by one letter in each case. The same holds for the string $z^{r1}B_j z^{l1}$. 
Consequently, we conclude that   the number of overlapped letters is bounded from above by four.
\\
\\
In case of $X_1  + Y_1  +  Z_1=1 $, we analyze exemplary the 
 constellation $(00,00,11)^s_j$. We set
$\psi_s(z)=1$,
$\psi_s(x)=0$ and $\psi_s(y)=0$. 
This  situation is displayed below.\\
\vspace{3mm}\\
\fbox{\parbox{\dimexpr \linewidth - 2\fboxrule - 2\fboxsep}{ 
\vspace{2mm}
$$
b
\medtriangleright
\medtriangleright 
\framebox{$x^{m0}$} \, 
\framebox{$x^{m0}$}  
 \medtriangleleft
\medtriangleleft 
m_1
\medtriangleright
\medtriangleright  
\framebox{$y^{m0}$} \, y^{r1}A_j y^{l1}
\framebox{$y^{m0}$}  
 \medtriangleleft
\medtriangleleft
m_2
\medtriangleright   
\medtriangleright
\framebox{$z^{r1}$} \,
 z^{r1}B_j z^{l1} \framebox{$z^{l1}$}
\medtriangleleft\medtriangleleft 
e
$$ 
$$
\downarrow
$$
$$
b~
\lefteqn{  \underbrace{ \phantom{ 
\medtriangleright
\medtriangleright  
\framebox{$x^{m0}$}   }}
_{
\medtriangleright
\medtriangleright  
\framebox{$x^{m0}$}  
}  }
\medtriangleright
\medtriangleright  
\framebox{$x^{m0}$} 
\medtriangleleft\medtriangleleft
m_1
\lefteqn{  \underbrace{ \phantom{ 
\medtriangleright
\medtriangleright \framebox{$y^{m0}$}   }}
_{\medtriangleright
\medtriangleright \framebox{$y^{m0}$} }  }
\medtriangleright
\medtriangleright  
\framebox{$y^{m0}$}
\medtriangleleft\medtriangleleft
m_2
\medtriangleright   
\medtriangleright
\lefteqn{  \underbrace{ \phantom{   \framebox{$z^{r1}$} B_j \framebox{$z^{l1}$}  }}
_{z^{r1}B_j z^{l1} }  }
\framebox{$z^{r1}$} B_j \framebox{$z^{l1}$}
\medtriangleleft\medtriangleleft e ~
y^{r1}A_j y^{l1}
$$ 
\vspace{2mm}
}}
\vspace{1mm}\\
Due to the $z^1$-alignment of the strings in $S^B(g^3_j)$, we obtain  an overlap of two letters.
Additionally, we align the string $\medtriangleright\medtriangleright  
\framebox{$x^{m0}$}$ from the left with $\framebox{$x^{m0}$} 
\medtriangleleft\medtriangleleft$. The same holds for 
$\medtriangleright
\medtriangleright \framebox{$y^{m0}$}$ and $\framebox{$y^{m0}$} 
\medtriangleleft\medtriangleleft$. Notice that it  is not more 
advantageous to align the string $x^{m0}B_jC_j$ with $\medtriangleright\medtriangleright  
\framebox{$x^{m0}$}$, since we lose the overlap of one letter with 
$\framebox{$x^{m0}$} \medtriangleleft\medtriangleleft$. Hence, we are able to get
an overlap of at most four letters, which corresponds to the satisfied 
equations $x\oplus x_{j_1+1}=0$, $y\oplus y_{j_2+1}=0$ and 
$z\oplus z_{j_3+1}=0$.\\
\\  
\textbf{CASE $X_1  \oplus Y_1  \oplus  Z_1=0 $ and $c$ is inconsistent:}\\
First, we concentrate on the constellations with the property  $X_1+Y_1+Z_1=2$. 
 Exemplary, we analyze the constellation
 $(0X_2,1Y_2,1Z_2)^s_j$ depicted below.\\
 \vspace{3mm}\\
\fbox{\parbox{\dimexpr \linewidth - 2\fboxrule - 2\fboxsep}{
\vspace{2mm}
 $$
 b
\medtriangleright
\medtriangleright  
\framebox{$x^{m0}$} \,
\framebox{$X_2$} 
\medtriangleleft\medtriangleleft
m_1
\medtriangleright
\medtriangleright  
 \framebox{$y^{r1}$} \,
y^{r1}A_j y^{l1} 
\framebox{$Y_2$}\medtriangleleft\medtriangleleft
m_2
\medtriangleright   
\medtriangleright 
 \framebox{$z^{r1}$} 
 z^{r1}B_j z^{l1}
\framebox{$Z_2$}\medtriangleleft\medtriangleleft
e
$$
$$
\downarrow
$$
 $$
 b
\medtriangleright
\medtriangleright  
\framebox{$x^{m0}$} \,
\framebox{$X_2$}
\medtriangleleft\medtriangleleft
m_1
\medtriangleright
\medtriangleright  
\lefteqn{  \underbrace{ \phantom{   \framebox{$y^{r1}$} A_j y^{l1}  }}
_{y^{r1}A_j y^{l1} }  }
\framebox{$y^{r1}$} A_j y^{l1}
\framebox{$Y_2$}\medtriangleleft\medtriangleleft
m_2
\medtriangleright   
\medtriangleright   
\lefteqn{  \underbrace{ \phantom{   \framebox{$z^{r1}$} B_j z^{l1}  }}
_{z^{r1}B_j z^{l1} }  }
\framebox{$z^{r1}$} B_j\, z^{l1}
\framebox{$Z_2$}\medtriangleleft\medtriangleleft
e
$$
\vspace{2mm}
}}
\vspace{1mm}\\
The strings $\medtriangleright
\medtriangleright  
 \framebox{$y^{r1}$}$ and 
 $\medtriangleright   
\medtriangleright 
 \framebox{$z^{r1}$}$
 can be used to  align from the right side with $z^{r1}B z^{l1}$
 and $y^{r1}A y^{l1}$, respectively. It yields an overlap of two letters. If the 
 corresponding equations with two variables are satisfied, which means
 $X_2=x^{m0}$, $Y_2=y^{l1}$ and $Z_2=z^{l1}$, we gain an overlap 
 of one letter per satisfied equation. Notice that using the $x^{0}$-alignment of $S(g^3_j)$
 does not yield more overlapped letters.
In summary, it is possible to attain an overlap of at most five letters, which  corresponds 
to the constellation $(00,11,11)^s_j$. 
An analogue argumentation holds for 
the constellations  $(1X_2,1Y_2,0Z_2)^s_j$
and $(1X_2,0Y_2,1Z_2)^s_j$.\\
\\
Next, we discuss constellations with the property $X_1+Y_1+Z_1=0$.
For this reason, we consider
 the constellation $(0X_2,0Y_2,0Z_2)^s_j$. \\
 \vspace{3mm}\\
\fbox{\parbox{\dimexpr \linewidth - 2\fboxrule - 2\fboxsep}{
\vspace{2mm}
$$
b
\medtriangleright
\medtriangleright
\framebox{$x^{m0}$}  \,
\framebox{$X_2$}
\medtriangleleft\medtriangleleft
m_1
\medtriangleright
\medtriangleright  
\framebox{$y^{m0}$}\,
y^{r1}B_jy^{r1}
 \framebox{$Y_2$}
\medtriangleleft
\medtriangleleft
m_2
\medtriangleright   
\medtriangleright
\framebox{$z^{m0}$} \,
y^{r1}B_jy^{r1}
 \framebox{$Z_2$}
\medtriangleleft\medtriangleleft
e
$$ 
$$
\downarrow
$$ 
$$
b
\medtriangleright
\medtriangleright
\lefteqn{  \underbrace{ \phantom{\framebox{$x^{m0}$} A C_j    }}
_{ x^{m0} A C_j  }  }
\framebox{$x^{m0}$}  A 
\lefteqn{  \overbrace{ \phantom{ C_jB  x^{m0}    }}
^{ C_jBx^{m0}  } }
C_jB  x^{m0} \,
\framebox{$X_2$}
\medtriangleleft\medtriangleleft
m_1
\medtriangleright
\medtriangleright  
\framebox{$y^{m0}$}\,
 \framebox{$Y_2$}
\medtriangleleft
\medtriangleleft
m_2
\medtriangleright   
\medtriangleright
\framebox{$z^{m0}$} \, \framebox{$Z_2$}
\medtriangleleft\medtriangleleft
e
$$ 
\vspace{2mm}
}}
 \vspace{1mm}\\
 Recall that $ x^{m0} C_j x^{m0}$ denotes the $x^0$-alignment of $S(g^3_j)$.
 This string can be aligned from the left
with $\medtriangleright
\medtriangleright x^{m0}$. If $X_2=x^{m0}$ holds, we achieve another
overlap of one letter. Furthermore, the string 
$\medtriangleright
\medtriangleright  
\framebox{$y^{m0}$} $ can be aligned from the right with $\framebox{$Y_2$}
\medtriangleleft
\medtriangleleft$ if and only if $Y_2=y^{m0}$ holds.
A similar argumentation can be applied to the strings
$\medtriangleright   
\medtriangleright
\framebox{$z^{m0}$}$ and $  \framebox{$Z_2$}
\medtriangleleft\medtriangleleft$. 
Finally, we note that we cannot benefit by aligning the string 
 $\framebox{$y^{l1}$}\medtriangleleft\medtriangleleft$ with 
 $ y^{r1} Ay^{l1}$. Consequently, we see that using the string 
 $x^{m0} C_j  x^{m0} $  is generally more profitable.
All in all, we gain an additional overlap of one letter
for satisfying $x\oplus y \oplus z =0$ and another 
overlap of one letter
if
the equation with two variables corresponding to the considered variable
is satisfied.\\ 
\\
\textbf{CASE $X_1  \oplus Y_1  \oplus  Z_1=1 $ and $c$ is inconsistent:}\\
Let us start with constellations satisfying $X_1+Y_1+Z_1=3$.
Exemplary, we analyze the constellation $(10,1Y_2,1Z_2)^s_j$. Due to the definition of $\psi_s$,
we set $\psi_s(x)=1-X_1$, $\psi_s(y)=1$ and $\psi_s(z)=1$. 
Notice that $\psi_s$  satisfies the equation $x\oplus y \oplus z=0$.
 By switching the value $\psi_s(x)$ from $X_1$ to $1-X_1$,
the equation $x_{j_1-1}\oplus x =0$ might become unsatisfied. 
Furthermore, we might lose an overlap of one letter 
by flipping the $1$-alignment of the 
strings corresponding to  $x_{j_1-1} \oplus x=0$ to the $0$-alignment. 
On the other hand, we gain an overlap of one letter
by aligning the string  $\medtriangleright
\medtriangleright 
\framebox{$x^{m0}$}$ from the right
side with $\framebox{$x^{m0}$}\medtriangleleft\medtriangleleft$.
This transformation yields at least one more satisfied equation.
In addition, the strings $y^{r1}A_j y^{l1} $ and $z^{r1}B z^{l1}$
can be aligned by one letter with 
$\medtriangleright
\medtriangleright  
 \framebox{$y^{r1}$}  $ and 
$\medtriangleright   
\medtriangleright
 \framebox{$z^{r1}$}$, respectively.
 If $Z_2=z^{l1}$ and $Y_2=y^{l1}$ holds, we achieve another overlap
 of one letter in each case. 
The situation is depicted below.\\
 \vspace{3mm}\\
\fbox{\parbox{\dimexpr \linewidth - 2\fboxrule - 2\fboxsep}
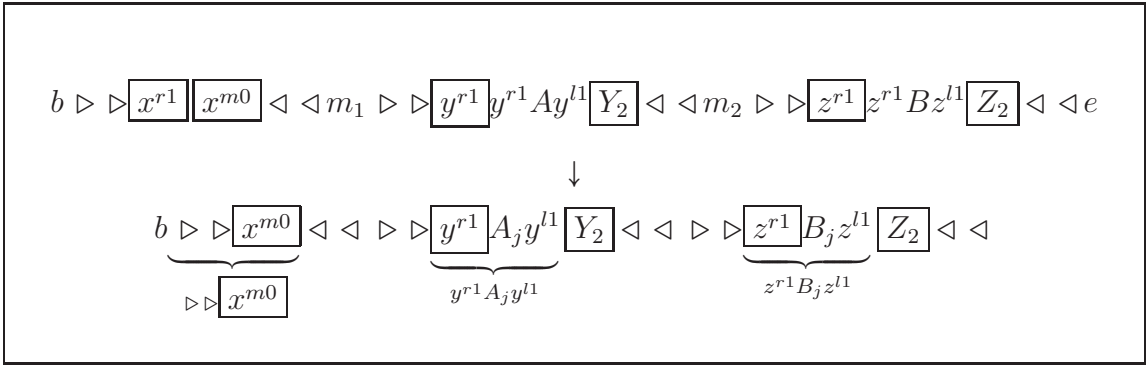
{
\vspace{4mm}
$$
b
\medtriangleright
\medtriangleright  
\framebox{$x^{r1}$} \,
\framebox{$x^{m0}$}  
\medtriangleleft\medtriangleleft
m_1
\medtriangleright
\medtriangleright  
 \framebox{$y^{r1}$}  
y^{r1}A y^{l1} 
\framebox{$Y_2$}
\medtriangleleft\medtriangleleft
m_2
\medtriangleright   
\medtriangleright
 \framebox{$z^{r1}$} 
z^{r1}B z^{l1} 
\framebox{$Z_2$}
\medtriangleleft
\medtriangleleft
e
$$ 
$$
\downarrow
$$
$$
b
\lefteqn{  \underbrace{ \phantom{
\medtriangleright
\medtriangleright 
\framebox{$x^{m0}$}    }}
_{\medtriangleright
\medtriangleright 
\framebox{$x^{m0}$} 
 }  }
\medtriangleright
\medtriangleright  
\framebox{$x^{m0}$}  
\medtriangleleft\medtriangleleft
\medtriangleright
\medtriangleright  
\lefteqn{  \underbrace{ \phantom{ \framebox{$y^{r1}$} A_j y^{l1}    }}
_{y^{r1}A_j y^{l1} }  }
\framebox{$y^{r1}$} A_j y^{l1} \,
\framebox{$Y_2$}
\medtriangleleft\medtriangleleft
\medtriangleright   
\medtriangleright
\lefteqn{  \underbrace{ \phantom{   \framebox{$z^{r1}$} B_j z^{l1}  }}
_{z^{r1}B_j z^{l1} }  }
\framebox{$z^{r1}$} B_j z^{l1} \,
\framebox{$Z_2$}
\medtriangleleft
\medtriangleleft
$$ 
\vspace{2mm}
}}
 \vspace{1mm}\\
The other constellations satisfying $X_1+Y_1+Z_1=3$ can be analyzed analogously.\\ 
\\
The remaining constellations $(X_1X_2,Y_1Y_2,Z_1Z_2)^s_j$ 
to be discussed 
 satisfy $X_1+Y_1+Z_1=1$ and are inconsistent.
Exemplary, we analyze the constellation $(01,0Y_2,1Z_2)^s_j$.
For $(01,0Y_2,1Z_2)^s_j$, 
 we set $\psi_s(x)=1-X_1$, $\psi_s(y)=Y_1$ and $\psi_s(z)=Z_1$.
 The scenario is depicted below.\\
 \vspace{3mm}\\
\fbox{\parbox{\dimexpr \linewidth - 2\fboxrule - 2\fboxsep}
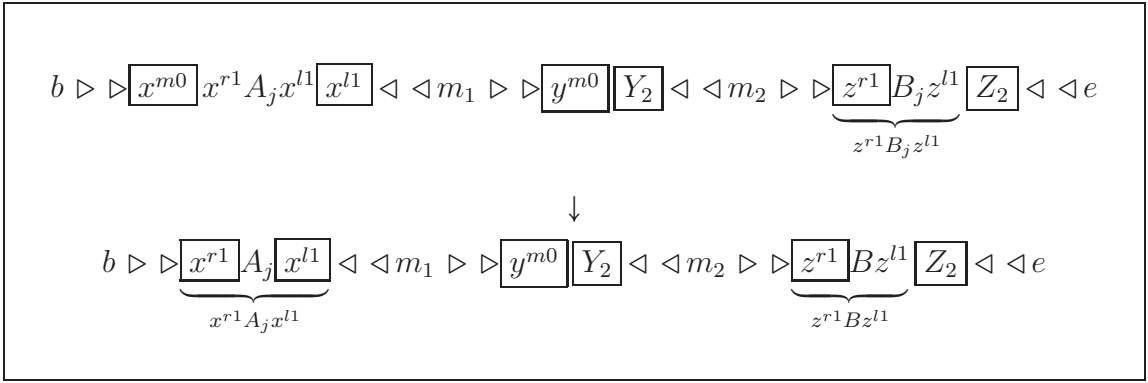
{
\vspace{2mm}
$$
b
\medtriangleright
\medtriangleright  
\framebox{$x^{m0}$} \,x^{r1}A_j x^{l1}
\framebox{$x^{l1}$}  
\medtriangleleft\medtriangleleft
m_1
\medtriangleright
\medtriangleright  
\framebox{$y^{m0}$} \,
\framebox{$Y_2$}
\medtriangleleft\medtriangleleft
m_2
\medtriangleright   
\medtriangleright
\lefteqn{  \underbrace{ \phantom{   \framebox{$z^{r1}$} B_j z^{l1}  }}
_{z^{r1}B_j z^{l1} }  }
\framebox{$z^{r1}$} B_j z^{l1} \,
\framebox{$Z_2$}
\medtriangleleft
\medtriangleleft
e
$$ 
$$
\downarrow
$$
$$
b
\medtriangleright   
\medtriangleright
\lefteqn{  \underbrace{ \phantom{   \framebox{$x^{r1}$} A_j \framebox{$x^{l1}$}  }}
_{x^{r1}A_j x^{l1} }  }
\framebox{$x^{r1}$} A_j \framebox{$x^{l1}$}
\medtriangleleft\medtriangleleft
m_1
\medtriangleright
\medtriangleright  
\framebox{$y^{m0}$}  \,
\framebox{$Y_2$}
\medtriangleleft\medtriangleleft
m_2
\medtriangleright   
\medtriangleright
\lefteqn{  \underbrace{ \phantom{   \framebox{$z^{r1}$} B z^{l1}  }}
_{z^{r1}B z^{l1} }  }
\framebox{$z^{r1}$} B z^{l1} \,
\framebox{$Z_2$}
\medtriangleleft
\medtriangleleft
e
$$ 
\vspace{2mm}
}}
 \vspace{1mm}\\
 By flipping the $0$-alignment of the 
strings corresponding to  $x_{j_1-1} \oplus x=0$ to the $1$-alignment, we can overlap $x^{r1}A_j x^{l1}$ 
from the left side with $\medtriangleright   
\medtriangleright\framebox{$x^{r1}$}$ and with 
$\framebox{$x^{l1}$}\medtriangleleft\medtriangleleft$
from the right side. This transformation achieves an overlap of at most one more
 letter. Moreover, we obtain at least one more satisfied
equation. If $Z_2=z^{l1}$ and
$Y_2=y^{m0}$ holds, it yields an overlap of three additional letters,
which corresponds to the constellation $(11,00,11)^s_j$.\\
\\
In summary, we note that it is possible to
 achieve  an overlap of at least one  letter in each case.
In addition to it, the assignment $\psi_s$ yields at least the same number
of satisfied equations as the number of overlapped letters.
 This means that if $\psi_s$ satisfies the equations $g^3_j$, 
$x \oplus x_{j_1+1} =0$, $y \oplus y_{j_2+1} =0$ and $z \oplus z_{j_3+1} =0$, the
corresponding strings in $s$ can have an overlap of at most five letters.

\subsection{Proof of  Theorem~\ref{hauptsatzsspI}}

Given an instance $\cH$ of the Hybrid problem  with $n$ circles,
$m_2$ equations with two variables and  $m_3$ equations with exactly
three variables with the properties described in Theorem~\ref{ssphybridsatz},
we construct in polynomial time an instance $\cS_{\cH}$ of
the Shortest Superstring problem with the  properties described in section~\ref{sspconstr}.
 Let $\phi$ be an assignment  to the variables of $\cH$ which leaves at most  $u$ equations
unsatisfied. According to section~\ref{sspgivenphi}, the length of the superstring $s_{\phi}$ is 
at most 
$$|s_{\phi}| \leq 7\cdot n+5\cdot m_2+  22\cdot m_3+u ,$$ 
since the length of the superstring increases by at most one letter for every unsatisfied equation
of the assignment.
Regarding the compression measure, we obtain the following.
\begin{eqnarray*}
comp(\cS_{\cH}, s_{\phi}) &\geq & 
\sum\limits_{s\in \cS_{\cH} } |s|-(7\cdot n+5\cdot m_2+  22\cdot m_3+u) \\
  &  =  &  (4+8)n+ 8\cdot m_2+ 36\cdot m_3-(7\cdot n+5\cdot m_2+  22\cdot m_3+u) \\
  &  =  &  5n+3m_2+14m_3-u
\end{eqnarray*}
\noindent 
On the other hand, given an superstring $s$ for $S_{\cH}$  with length 
$$|s|=5m_2+22 m_3+u+7n,$$
we can construct in polynomial time an normed superstring $s'$ without increasing
the length of it by applying the transformations defined in section~\ref{sspassignment}.
This enables us to define an  
 assignment $\psi_s$ to the variables
of $\cH$ according section~\ref{sspassignment} that leaves at most $u$ equations in $\cH$ unsatisfied.
A similar argumentation leads to the conclusion that given a superstring $s$ for $S_{\cH}$  with
compression 
$$comp(\cS_{\cH}, s_{\phi})= 5n+3m_2+14m_3-u,$$
 we construct in polynomial time
an assignment to the variables in $\cH$ such that at most $u$ equations are unsatisfied. 
\qed \\
\\
Next, we are going to describe smaller gadgets for equations with three variables
implying an improved explicit lower bound and give the proof of
 Theorem~\ref{hauptsatzssp}. 

\subsection{Proof of  Theorem~\ref{hauptsatzssp}}
\label{sec:improvapproach}
Given an equation with three variables $g^3_c\equiv x\oplus y \oplus z =0$, 
we introduce the sets $S^\alpha(g^3_j)$ and $S^\beta(g^3_j)$ including the following 
strings. 
$$
x^{r1\alpha}x^{l1}y^{r1}y^{l1}, \qquad y^{r1}y^{l1} x^{m0} C_j, 
\qquad 
x^{m0} C_j x^{r1\alpha}  x^{l1} \qquad \in S^\alpha(g^3_j)
$$
$$
x^{r1\beta} x^{l1}    z^{r1} z^{l1}, \qquad 
z^{r1}   z^{l1} C_j  x^{m0}, \qquad 
C_j  x^{m0}  x^{r1\beta}   x^{l1} \qquad \in S^\beta(g^3_j)
$$
In addition, we introduce new strings for the equation $x_{i-1}\oplus x=0$.
On the other hand, 
the strings corresponding to $x\oplus x_{i+1}=0$, 
$y_{i-1}\oplus y=0$, $y\oplus y_{i+1}=0$, $z_{i-1}\oplus z=0$
and $z\oplus z_{i+1}=0$ remain the same.
Let us define the strings for $x_{i-1}\oplus x=0$:
$$
x^{l1}_{i-1} x^{r1\beta}   x^{l1}_{i-1} x^{r1\alpha}
\quad
x^{l1}_{i-1} x^{r1\alpha} x^{m0}_{i-1} x^{m0} 
\quad 
x^{m0}_{i-1} x^{m0}  x^{l1}_{i-1} x^{r1\beta}
$$
These three strings can be aligned each by two letters in a cyclic fashion.
Accordingly, we obtain three combinations that can be used to overlap with other
strings by one letter from the left side as well as from the right side. 
Note that we have only two combinations if we consider only the left most position
of the combined strings. For example, the combination 
 $$
x^{l1}_{i-1} x^{r1\beta} 
 x^{l1}_{i-1} x^{r1\alpha} 
x^{m0}_{i-1} x^{m0}  x^{l1}_{i-1} x^{r1\beta}
$$
 can be used to overlap from the right side with strings in $S^\beta(g^3_j)$, whereas
$$
x^{l1}_{i-1} x^{r1\alpha} 
x^{m0}_{i-1} x^{m0}  
x^{l1}_{i-1} x^{r1\beta}
x^{l1}_{i-1} x^{r1\alpha} 
$$
can be aligned with strings contained in $S^\alpha(g^3_j)$.
Therefore, we may apply the same arguments as in the proof of Theorem~\ref{hauptsatzssp}.
The strings corresponding to equations of the form $x\oplus y\oplus z =1$
can be constructed analogously.\\
\\
Given an instance $\cH$ of the Hybrid problem  with $n$ circles,
$m_2$ equations with two variables and  $m_3$ equations with exactly
three variables with the properties described in Theorem~\ref{ssphybridsatz},
we construct in polynomial time an instance $\cS_{\cH}$ of
the Shortest Superstring problem.
 Let $\phi$ be an assignment  to the variables of $\cH$ which leaves at most  $u$ equations
unsatisfied. Then, it is possible to construct a superstring $s_{\phi}$ with length 
$$
|s_{\phi}| \leq 7\cdot n+5\cdot m_2+  16\cdot m_3+u ,
$$ 
since the length of the superstring increases by at most one letter for every unsatisfied equation
of the assignment. 
Regarding the compression measure, we obtain the following.
\begin{eqnarray*}
comp(\cS_{\cH}, s_{\phi}) &\geq & 
\sum\limits_{s\in \cS_{\cH} } |s|-(7\cdot n+5\cdot m_2+  16\cdot m_3+u) \\
  &  =  &  (4+8)n+ 8\cdot m_2+ 28\cdot m_3-(7\cdot n+5\cdot m_2+  16\cdot m_3+u) \\
  &  =  &  5n+3m_2+12m_3-u
\end{eqnarray*}
On the other hand, given an superstring $s$ for $S_{\cH}$  with length 
$$|s|=5m_2+16 m_3+u+7n,$$
we can construct in polynomial time a normed superstring $s'$ without increasing
the length of it.
The corresponding
 assignment $\psi_{s'}$ to the variables
of $\cH$  leaves at most $u$ equations in $\cH$ unsatisfied.
A similar argumentation leads to the conclusion that given a superstring $s$ for $S_{\cH}$  with
compression 
$$comp(\cS_{\cH}, s_{\phi})= 5n+3m_2+12m_3-u,$$
 we construct in polynomial time
an assignment to the variables in $\cH$ such that at most $u$ equations are unsatisfied. 
\qed 
\section{ Concluding Remarks }
It seems that a new method
is needed now in order to obtain better approximation lower bounds. 
Perhaps direct PCP constructions are the natural next step for proving stronger
approximation hardness results for the problems considered in this paper.


\begin{thebibliography}{alpha}
%
\bibitem[AS95]{AS95}
C. Armen and C. Stein,
\emph{Improved Length Bounds for the Shortest Superstring Problem}, 
Proc. 5th WADS (1995), LNCS 955, 1995, pp. 494--505.
%
\bibitem[AS98]{AS98}
C. Armen, C. Stein,
\emph{ A $2\frac 23$ Superstring Approximation Algorithm},
Discrete Applied Mathematics 88 (1998), pp. 29--57.



\bibitem[BK99]{BK99}
P. Berman, M. Karpinski,
\emph{ On Some Tighter Inapproximability Results}, 
Proc. 26th ICALP (1999), LNCS 1644, 1999, pp. 200--209.


\bibitem[B02]{B02} 
M. Bl\"aser,
\emph{ An $\frac{8}{13}$ Approximation Algorithm for
the Asymmetric Max-TSP},  Proc.  13th SODA (2002), 64--73.
%

%
%
\bibitem[BJL$^+$94]{BJL$^+$94}
A. Blum, T. Jiang, M. Li, J. Tromp, and M. Yanakakis,
\emph{Linear Approximation of Shortest Superstrings},
J. ACM 41 (1994), pp. 630--647.
%
\bibitem[BJJ97]{BJJ97}
D. Breslauer, T. Jiang, and Z. Jiang, 
\emph{Rotations of Periodic Strings and Short Superstrings},
J. Algorithms 24 (1997), pp. 340--353.
%
\bibitem[CGP$^+$94]{CGP$^+$94}
A. Czumaj, L. Gasieniec, M. Piotrow, and W. Rytter,
\emph{ Parallel and Sequential Approximations
of Shortest Superstrings}, 
Proc. 1st SWAT (1994), LNCS 824, 1994, pp. 95--106.
%
%
\bibitem[E99]{E99}
L. Engebretsen, 
\emph{An Explicit Lower Bound for TSP with
Distances One and Two}, 
Algorithmica 35 (2003), pp. 301--318.
%
\bibitem[EK06]{EK06}
L. Engebretsen and M. Karpinski,
\emph{ TSP with Bounded Metrics},
 J. Comput. Syst. Sci. 72 (2006), pp. 509--546.
%
%
\bibitem[FNW79]{FNW79} 
  M.~Fisher, G.~Nemhauser and L.~Wolsey,
  \emph{An Analysis of Approximations for Finding a Maximum Weight
Hamiltonian Circuit}, 
Oper. Res. 27 (1979), pp. 799--809.

\bibitem[GMS80]{GMS80}
J. Gallant, D. Maier and  J.~Storer,
\emph{ On Finding Minimal Length Superstrings}
J. Comp. Sys. Sci. 20 (1980), pp. 50--58.

\bibitem[H01]{H01}
J. H{\aa}stad,
\emph{ Some Optimal Inapproximability Results},
 J. ACM 48 (2001), pp. 798--859.

\bibitem[KLS$^+$05]{KLS$^+$05}
  H. Kaplan, M. Lewenstein, N. Shafrir and  M. Sviridenko,
  \emph{Approximation Algorithms for Asymmetric TSP by Decomposing Directed Regular 
  Multigraphs}, J. ACM 52 (2005), pp. 602--626.

\bibitem[KPS94]{KPS94}
R. Kosaraju, J. Park, and C. Stein, 
\emph{Long Tours and Short Superstrings}, Proc. 35th FOCS (1994), 
 pp. 166--177.

\bibitem[L88]{L88}
A.~Lesk,
\emph{Computational Molecular Biology, Sources and Methods for Sequence Analysis},
Oxford University Press, 1988.

 


\bibitem[LS03]{LS03}
M. Lewenstein and M. Sviridenko,
\emph{ Approximating Asymmetric Maximum TSP},
Proc. 14th SODA (2003), pp. 646--654.

\bibitem[L90]{L90}
M.~Li,
\emph{Towards a DNA Sequencing Theory},
 Proc. 31th FOCS (1990), pp. 125--134. 

 




\bibitem[MS77]{MS77}
  D. Maier and J.~Storer,
  \emph{A Note on the Complexity of the Superstring Problem},
  Report No. 223, Computer Science Laboratory, Princeton University,
  1977.
	
	\bibitem[MJ75]{MJ75} 
A. Mayne and E.~James, 
\emph{Information Compression
by Factorising Common Superstrings},
Comput. J. 18 (1975), pp.  157--160.
	
	\bibitem[M12]{M12}
	M.~Mucha,
	\emph{Lyndon Words and Short Superstrings},
	CoRR abs/1205.6787, 2012.

	
	\bibitem[M94]{M94}
  M. Middendorf,
   \emph{ More on the Complexity of Common Superstring and Supersequence Problems},
    Theor. Comput. Sci. 125 (1994), pp. 205--228.
  
\bibitem[M98]{M98}
  M. Middendorf,
   \emph{ Shortest Common Superstrings and Scheduling with Coordinated Starting Times},
   Theor. Comput. Sci. 191 (1998), pp. 205--214.  

\bibitem[O99]{O99}
S. Ott,
\emph{Lower Bounds for Approximating Shortest Superstrings over an Alphabet
of Size $2$}, Proc. 25th. WG (1999), LNCS 1665, pp. 55--64. 

\bibitem[PY93]{PY93}   
C.~Papadimitriou and M. Yannakakis,
\emph{ The traveling
Salesman Problem with Distances One and Two},
 Math. Oper. Res. 18 (1993), pp. 1--11.

\bibitem[S88]{S88}
J.~Storer,
\emph{ Data Compression: Methods and
Theory}, Computer Science Press, 1988.

\bibitem[SS82]{SS82}
J.~Storer and T.~Szymanski,
\emph{Data Compression
via Textual Substitution}, 
J. ACM 29 (1982), pp. 928--951.

\bibitem[S99]{S99}
Z. Sweedyk, 
\emph{A $2\frac 12$-Approximation Algorithm for Shortest Superstring},
 SIAM J. Comput. 29 (1999), pp. 954--986.

\bibitem[TU88]{TU88}
J. Tarhio and E. Ukkonen,
\emph{A Greedy Approximation Algorithm
for Constructing Shortest Common Superstrings},
Theor. Comput. Sci. 57 (1988), pp. 131--145.  

\bibitem[TT93]{TT93}
S. Teng and F. Yao,
\emph{ Approximating shortest superstrings},
 Proc. 34th FOCS (1993), pp. 158--165, 1993.  

\bibitem[T90]{T90}
V.~Timkovskii, 
\emph{Complexity of Common 
Subsequence and Supersequence Problems 
and Related Problems},
Cybernetics and Systems Analysis
25 (1990), pp. 565--580;
translated from Kibernetika 25 (1989),
pp. 1--13.

\bibitem[T89]{T89}
J. Turner,
\emph{ Approximation Algorithms
for the Shortest Common Superstring problem },
Information and Computation 83 (1989), pp. 1--20.


  

  
  
\bibitem[V05]{V05} 
  V. Vassilevska,
   \emph{Explicit Inapproximability Bounds for the Shortest Superstring Problem},
    Proc. 30th MFCS (2005), LNCS 3618, 2005, pp. 793--800.
  
\bibitem[V92]{V92}
S. Vishwanathan, 
\emph{ An Approximation Algorithm for the Asymmetric Travelling Salesman Problem
 with Distances One and Two}, 
 Inf. Process. Lett. 44 (1992), pp. 297--302.   
  
  




\end{thebibliography}
\end{document}